\documentclass[iop,revtex4,apj,onecolumn]{emulateapj}
\usepackage{bm}
\usepackage{graphics,epsfig,rotate,lscape,graphicx,amsmath,amsthm,amssymb}
\usepackage{float,amsfonts,amsbsy,hyperref,delarray,amsfonts,amscd,pifont}
\usepackage{color,multirow}
\usepackage{algorithmicx}
\usepackage{algorithm}
\usepackage{enumerate}

\usepackage{natbib}

\newtheorem{claim}{Claim}

\newtheorem{lemma}{Lemma}

\newcommand{\re}{\text{\,\rm re}}
\newcommand{\im}{\text{\,\rm im}}
\newcommand{\bs}{\boldsymbol}

\begin{document}

\title{Bayesian inference of CMB gravitational lensing}
\shorttitle{Bayesian CMB lensing}
\shortauthors{E. Anderes, B. D. Wandelt \& G. Lavaux}

\author{Ethan Anderes\altaffilmark{1}}
\affil{Department of Statistics, University of California, Davis CA 95616, USA.}

\author{Benjamin D. Wandelt\altaffilmark{2}}
\author{Guilhem Lavaux\altaffilmark{2}}

\affil{Sorbonne Universit\'es, UPMC Univ Paris 06 \& CNRS, UMR7095, Institut d'Astrophysique de Paris, F-75014, Paris, France}

\begin{abstract} 
The Planck satellite, along with several ground based telescopes, have mapped the  cosmic microwave background (CMB) at sufficient resolution and signal-to-noise so as to allow a detection of the subtle distortions  due to the gravitational influence of the intervening  matter distribution. A natural modeling approach is to write a Bayesian hierarchical model for the lensed CMB in terms of the unlensed CMB and the lensing potential. So far there has been no feasible algorithm for inferring the posterior distribution of the lensing potential from the lensed CMB map. We propose a solution that allows efficient Markov Chain Monte Carlo sampling from the joint posterior of the lensing potential and the unlensed CMB map using the Hamiltonian Monte Carlo technique. The main conceptual step in the solution is  a re-parameterization of CMB lensing in terms of the lensed CMB and the ``inverse lensing'' potential. We demonstrate a fast implementation on simulated data including noise and a sky cut, that uses a further acceleration based on a very mild approximation of the inverse lensing potential.  We find that the resulting Markov Chain has short correlation lengths and excellent convergence properties, making it promising for application to high resolution CMB data sets of the future.
\end{abstract}


\section{Introduction}

 Over the past few years, data from ground based telescopes (ACT, SPT, Polarbear) and the Planck satellite have resulted in an unprecedented detection of weak gravitational lensing of the cosmic microwave background (CMB) \citep{das2011detection,van2012measurement,planck2013lensing,Polarbear2014,planck2015lensing}.  Upcoming high resolution, high signal-to-noise experiments are poised to    make the gravitational lensing distortion  a  powerful probe of cosmology, dark matter, and neutrino physics. The state-of-the-art estimator of CMB gravitational lensing, the quadratic estimator developed by Hu and Okomoto \citep{hu2001mapping,hu2002mass}, works in part through a delicate cancellation of terms in an infinite Taylor expansion of the lensing effect on the CMB. The effect of this cancellation is particularly sensitive to foreground contaminants and sky masking, which  if not fully accounted for,  limits  the statistical inferential power of this new data.  

 Possibly the most promising alternative to the quadratic estimator is Bayesian lensing. It has been known for some time that the quadratic estimator is suboptimal for high signal-to-noise, high resolution experiments and that a full Bayesian treatment can overcome this limitation
 \citep{HirataSeljak1, HirataSeljak2}. Indeed, Bayesian techniques applied to the lensed CMB observations have the potential to drastically change the way lensing is estimated and used for inference.  Current frequentist estimators of the unknown lensing potential treat the unlensed CMB as a source of shape noise which is marginalized out. Conversely, a Bayesian lensing posterior treats the lensing potential {\em and} the unlensed CMB as joint unknowns, whereby obtaining scientific constrains jointly rather than marginally. Moreover, the  posterior distribution is easier to interpret and sequentially update with additional data. From the geometry of weak lensing, most of the lensing power comes from matter at high redshift $z\sim 2$. At these distances the matter distribution on large scales is well approximated by Gaussian density fluctuations. In addition, the unlensed CMB is, at present, indistinguishable from an isotropic Gaussian random field.  From a statistical perspective, this is a perfect scenario for Bayesian methods in that both the observations and the unknown lensing potential are {\em physically predicted} to be Gaussian random fields.
 
Physicists have known, for some time, that Bayesian methods could potentially provide next-generation lensing estimates. In their seminal review \cite{Lewis20061} discuss the possibility of obtaining posterior draws from the lensing potential and the unlensed CMB jointly. However, they acknowledge the main obstacle for naive Gibbs implementations:
\begin{quote}
``... given a particular lensing potential the delensed sky is given essentially by a delta function.
This means that naive Gibbs iterations will not converge within a reasonable time. At the time of writing there are no known practical methods for sampling from the full posterior distribution."
\end{quote}
In this paper we show that, indeed, there does exist a practical way to obtain Gibbs iterations which converge quickly. The solution is through a re-parameterization of CMB lensing problem. Instead of treating the lensing potential as unknown we work with inverse-lensing or what we call anti-lensing. Surprisingly, the slowness of naive Gibbs translates to fast convergence of the re-parameterized Gibbs chain.

In Section \ref{two parameter system} we motivate our re-parameterization by analyzing a simple two parameter statistical problem.  The concepts are then applied to the Bayesian lensing problem in Section \ref{Section: Ancillary and sufficient parameters for the lensed CMB}. The two conditional distributions in our Gibbs implementation are discussed in Section \ref{Section: hamiltonian sampler section} and Section \ref{Section: iterative message passing section}. We finish with some  simulation examples in Section \ref{Section: simulation examples}.

All the code presented in this paper is written in the language {\em Julia} \citep{bezanson2012julia} and is publicly available through the on-line repository \url{https://github.com/EthanAnderes/BayesianCmbLensing}.

%
%
\section{Weak lensing primer and a Bayesian challenge}
\label{primer}

The effect of weak lensing is to simply remap the CMB, preserving surface brightness.   Up to leading order, the remapping  displacements are given by $\nabla \phi$, where $\phi$ denotes the lensing potential and is the planar projection of the three dimensional gravitational potential (see \cite{dodelson2003modern}, for example). Therefore the lensed CMB can  be written $T(x + \nabla \phi(x))$ where $T(x)$ denotes the unlensed CMB temperature fluctuations and $x$ represents an observational direction on the unit sphere. For this paper we will be focusing on the small angle limit  so that $x$ is assumed to vary in a small patch of $ \Bbb R^2$. However, we do not expect the fast convergence properties of our algorithm to be sensitive to the small angle approximation and the methodology presented here should hold for a full treatment on the sphere.  The lensed CMB is observed with additive noise (denoted $n(x)$) to result in data of the form
\begin{equation}
\label{data model}
 \text{data}(x)= T(x + \nabla \phi(x))+ n(x).
\end{equation}
The goal of weak lensing surveys is to use the data in (\ref{data model}) to  estimate $\phi$, $T$  and possibly  the spectral densities of $T$ and $\phi$.  

A natural approach to develop a Bayesian lensing estimator is to generate posterior samples through a Gibbs algorithm which iteratively samples from the two conditionals: $P(T |  \phi,\text{data})$ and $P(\phi | T,  \text{data})$.
Sampling from $P(T |  \phi,\text{data})$ is simply a Gaussian random field prediction problem since conditioning on $\phi$ models the data as
\[
\text{data}(x) = T(\!\!\!\underbrace{x+\nabla\phi(x)}_\text{\tiny known obs locations}\!\!\!) + n(x).
\]
In other words, the data is a noisy version of  $T$ observed on an irregular grid. 
Conversely, when sampling from $P(\phi | T,\text{data})$ the data is of the form
\[
\text{data}(x) = \underbrace{\!\!T\!\!}_{\text{\tiny known}}(x+\nabla\phi(x)) + n(x). 
\]
To see how one might approximate this conditional notice first that the CMB field $T(x)$ is very smooth. Indeed, Silk damping  predicts  a exponentially decaying power spectrum. Therefore a linear Taylor approximation,  $\text{data}(x) \approx T(x) + \nabla T(x)\cdot \nabla\phi(x) + n(x)$, may be useful. In fact, the derivation of the quadratic estimator explicitly uses this linear approximation. 
If one is willing to use this linear approximation then the conditional $P(\phi | T,\text{data})$ is simply a Bayesian regression problem since $T$ (and thus $\nabla T$) are both known with a Gaussian prior on $\nabla \phi$.

Unfortunately, the structure of both of these conditionals make the Gibbs very slow to converge. The case is exacerbated in the situation when noise level is small. For example, in the second conditional, if $T$ is known and fixed, the extent of the likely $\phi$'s under $P(\phi|T,\text{data})$  is very small compared to the likely $\phi$'s under $P(\phi, T| \text{data})$. 
This suggests a highly dependent posterior $P(\phi, T| \text{data})$.

%
%
\section{Two parameter analogy}
\label{two parameter system}

To motivate our solution to the Bayesian lensing problem we start with a simple two parameter statistical problem.  This system has two unknown parameters $ t, \varphi$ with a single data point given by
\[\text{data} =  t + \varphi + n\]
where $n$ denotes additive noise.  In the Bayesian setting, the posterior distribution is computed as 
\begin{equation}
\label{post1}
 P( t,\varphi|\text{data})\propto P(\text{data}| t, \varphi) P( t,\varphi) 
 \end{equation}
where $P(\text{data}| t, \varphi)$ denotes the likelihood of the data given  $ t, \varphi$ and $P( t,\varphi)$ denotes the  prior on $ t, \varphi$. 
The Gibbs sampler is a widely used algorithm for generating (asymptotic) samples from  $P( t, \varphi|\text{data})$. The algorithm generates a Markov chain of parameter values $( t^{1}, \varphi^{1}), ( t^{2}, \varphi^{2}),\ldots$ generated by iteratively sampling from the conditional distributions:
\begin{align*}
 t^{i+1} &\sim P( t|\varphi^{i},\text{data}) \\
\varphi^{i+1}   &\sim P(\varphi| t^{i+1},\text{data}).
\end{align*}
A useful heuristic for determining the convergence rate of a Gibbs chain is the extent to which the two parameters $ t$ and $\varphi$ are dependent in $P(t, \varphi|\text{data})$. A highly dependent posterior $P( t, \varphi|\text{data})$ leads to a slow Gibbs chain, near independence leads to a fast Gibbs chain. Indeed, exact independence gives a sample of the posterior after one Gibbs step.  A technique for accelerating the convergence of a Gibbs sampler is to find a  re-parameterization of $ t$ and $\varphi$ in a way which makes the posterior less dependent. In the remainder of this section we discuss a specific re-parameterization which, by analogy, can be applied to Bayesian lensing.

The relevant situation for Bayesian lensing is the case that $ t$ and $\varphi$ are highly negatively correlated in $P( t, \varphi|\text{data})$.  This motivates re-parameterizing $( t,\varphi)$ to $(\widetilde  t, \varphi)$ where $\widetilde  t \equiv  t + \varphi$ so that
\begin{align*}
\text{data} &= \widetilde  t + n.
\end{align*}
In the statistics literature,  $( t, \varphi)$ has been referred to as an {\bf ancillary parameterization} whereas $(\widetilde  t, \varphi)$ is referred to as a {\bf sufficient parameterization}.
We note that the terms {\em ancillary} and {\em sufficient } parameterization have been used interchangeably with the nomenclature {\em non-centered} and {\em centered } parameterizations, respectively, in the statistics literature \cite{bernardo2003non,gelfand1995efficient,papaspiliopoulos2008stability,papaspiliopoulos2007general,yu2011center}.  Figure  \ref{fastslowGibbs} illustrates the difference between an ancillary versus sufficient posterior distribution for our simple two parameter model. The left plot shows the posterior density contours for the ancillary parameterization $( t, \varphi)$, along with 20 steps of a Gibbs sampler.  Conversely, the right plot shows the posterior density contours for the sufficient chain $(\widetilde  t, \varphi)$ with 20 Gibbs steps. Notice that negative correlation  in the ancillary parameterization manifests in near independence for the sufficient chain.  Indeed, the slower the ancillary chain the faster the sufficient chain and vice-versa.

\begin{figure*}
\begin{center}
\includegraphics[height=2.0in]{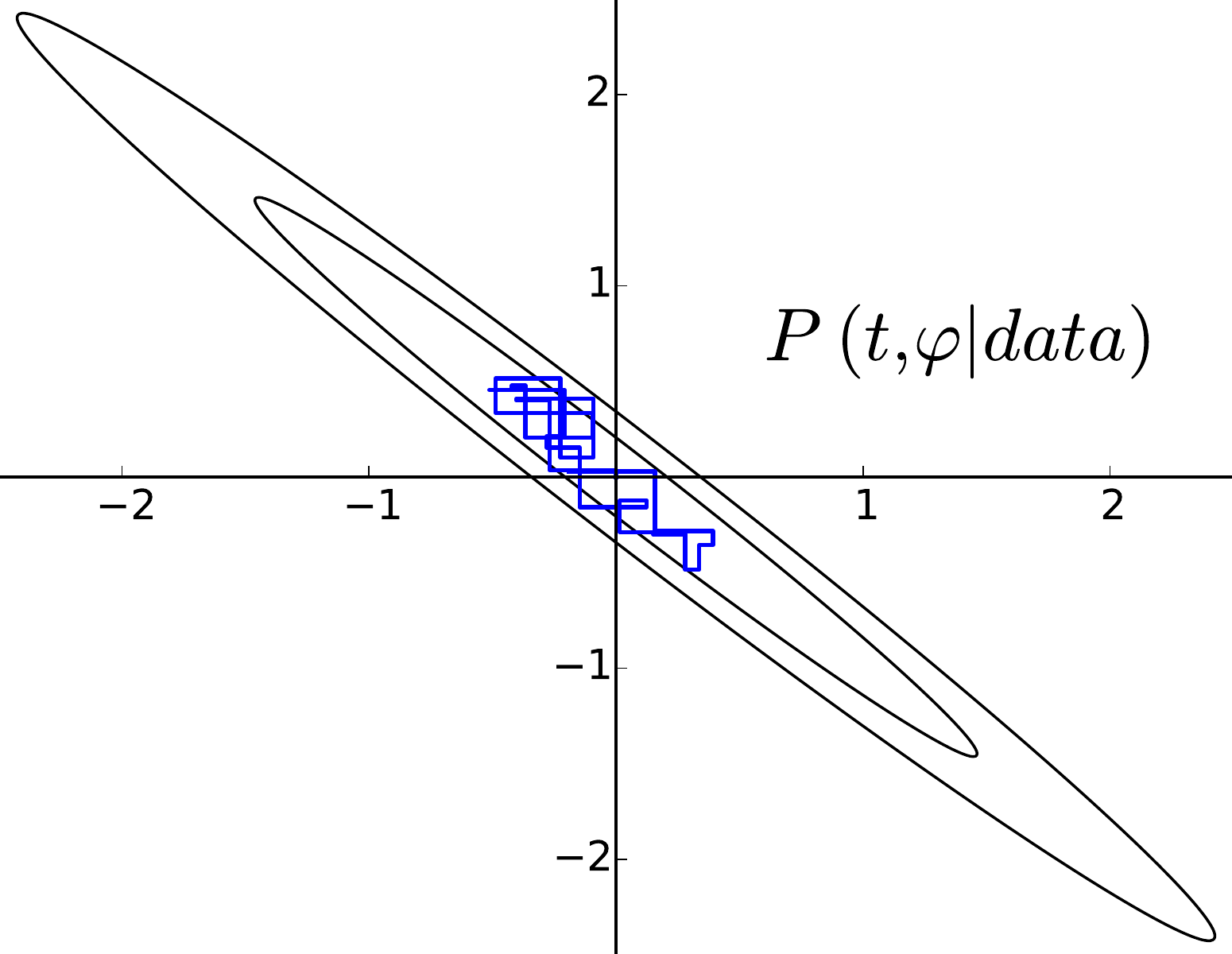}
\includegraphics[height=2.0in]{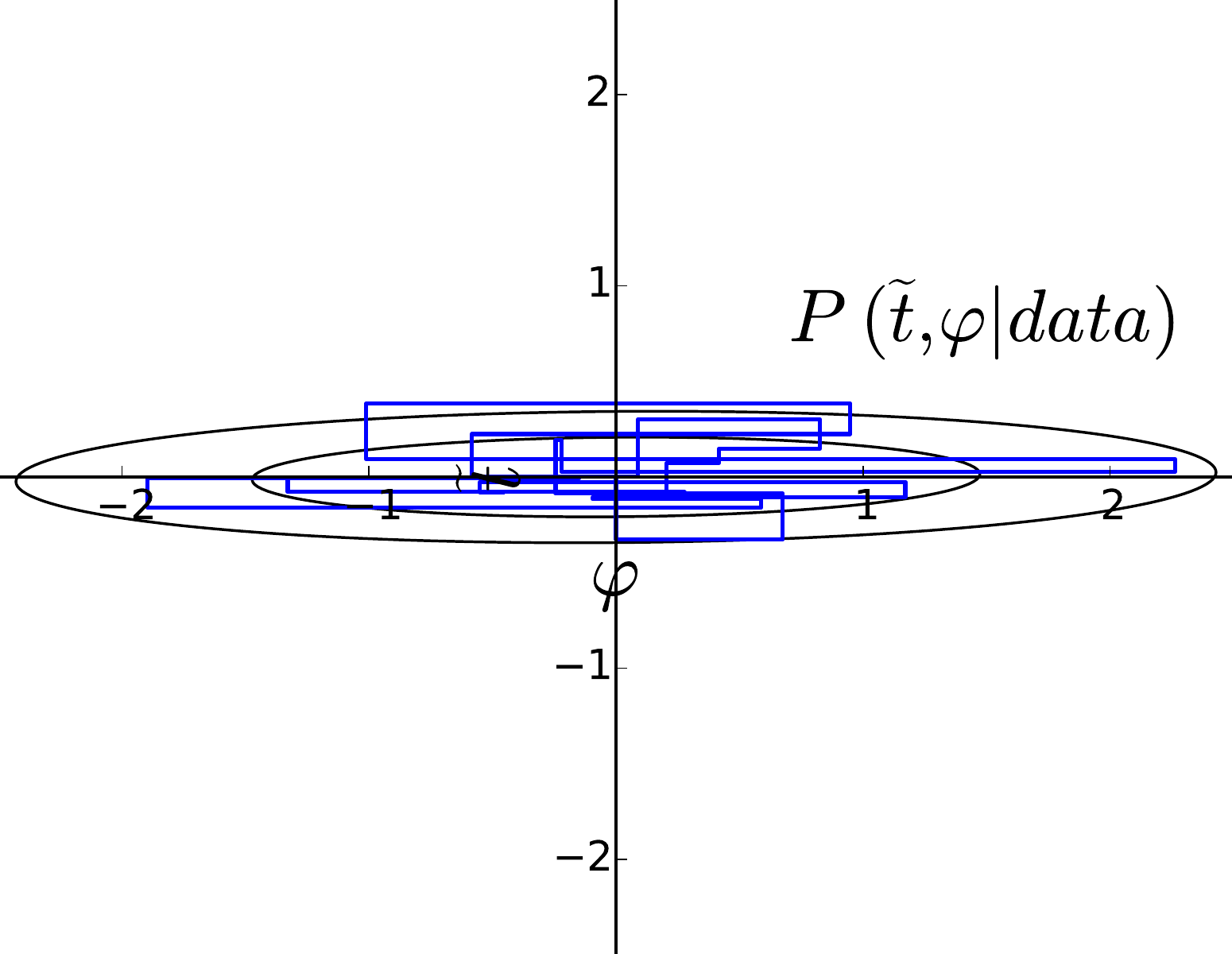}
\end{center}
\caption{\label{fastslowGibbs}
{\em Left:} density contours of the {\bf ancillary} chain $P( t, \varphi|\text{data})$ with 20 steps of a Gibbs sampler. {\em Right:} density contours of the {\bf sufficient}  chain $P(\widetilde  t, \varphi|\text{data})$ with 20 steps of a Gibbs sampler. This illustrates the general heuristic that a slowly converging ancillary chain translates to a quickly converging the sufficient chain.}
\end{figure*}

%
%
\section{Ancillary versus sufficient parameters for the lensed CMB}
\label{Section: Ancillary and sufficient parameters for the lensed CMB}

The ancillary parameterization presented in the previous section is analogous to the lensed CMB problem as follows
\begin{align*}
  \text{data}(x) &= T(x+\nabla \phi(x)) + n(x)\quad\text{\bf\em analogous to}\quad
  \text{data} =  t +\varphi + n
\end{align*}
where the unlensed  CMB temperature field $T$ and the lensing potential $\phi$ are the two unknown parameters. As was discussed in Section \ref{primer} the Gibbs chain based on the ancillary parameters $T(x)$ and $\phi(x)$ is exceedingly slow.  This clearly motivates the following re-parameterization to sufficient parameters for the lensed CMB problem 
\begin{align*}
  \text{data}(x) &= \widetilde T(x) + n(x) \quad\text{\bf\em analogous to}\quad
  \text{data} =  \widetilde t + n
\end{align*}
where now $\widetilde T$ denotes the lensed CMB temperature field with no noise or beam.
The sufficient chain then proceeds as
\begin{align}
\label{suff 1} \widetilde T^{i+1}&\sim P(\widetilde T |  \phi^{i},\text{data}) \\
\label{suff 2} \phi^{i+1}&\sim P(\phi | \widetilde T^{i+1},  \text{data}).
\end{align}
In Section \ref{Section: iterative message passing section} we adapt an iterative message passing algorithm, originally developed in \cite{elsner2013efficient,jasche2014matrix}, for Wiener filtering and sampling from (\ref{suff 1}). In Section \ref{Section: hamiltonian sampler section} we derive a Hamiltonian Markov Chain algorithm to sample from (\ref{suff 2}). Our Hamiltonian Markov Chain algorithm relies on an approximation---motivated again by the two parameter system---we call {\em anti-lensing}.

%
%
\subsection{Anti-lensing approximation}
\label{section: Anti-lensing approximation}

In the two parameter analogy from Section \ref{two parameter system},  the relation between the sufficient parameter $\widetilde t$ and the ancillary parameter $t$ is given by $\widetilde t - \varphi = t$. The corresponding relation for CMB lensing we refer to as {\em anti-lensing}:
\begin{equation}
\label{antilense}
 \widetilde T(\,\underbrace{x-\nabla\phi(x)}_\text{\small anti-lensing}\,) \approx T(x).
\end{equation}
We distinguish between {\em inverse lensing} and  {\em anti-lensing}. Inverse lensing denotes the true coordinate displacement which, when applied to  $\widetilde T$, recovers the unlensed $T$. Conversely, anti-lensing is given by $-\nabla\phi$ and approximates inverse lensing.  

To examine the difference between anti-lensing and inverse lensing notice that an extra divergence-free potential is needed to model the inverse lensing displacement field. 
Indeed, let $f(x):=x + \nabla \phi(x)$ denote the lensing map. With this notation we have
\[
\widetilde T(x) = T(f(x)) \quad\text{ and }\quad T(x) = \widetilde T(f^{-1}(x))
\]
where $f^{-1}$ is the inverse lensing map that satisfies $x = f^{-1}(f(x))$. Now let $d(x)$ denote the displacement vector field for inverse lensing so that $f^{-1}(x) = x + d(x)$. Therefore
$x = f^{-1}(f(x)) = f(x) + d(f(x))$.
In particular $d(f(x)) = x - f(x) =   -\nabla \phi(x)$ which gives
\[ d(x) = -\nabla \phi(f^{-1}(x)). \]
This implies that the inverse lensing displacement is modeled as a warped version of the curl-free vector field $-\nabla \phi$ (warped by $f^{-1}$).  This warping introduces a non-zero divergence-free term (just as lensing adds non-zero B-mode power in the CMB polarization).

To illustrate the expected magnitudes of the divergence-free and curl-free terms, start with a Helmholtz decomposition of the inverse lensing displacement: $d(x) = -\nabla \phi^\text{inv}(x) - \nabla^\perp \psi^\text{inv}(x),$  where $\nabla^\perp \equiv \bigr(-\frac{\partial}{\partial y},\frac{\partial}{\partial x} \bigl)$ and  $\psi^\text{inv}$ denotes a stream function potential which models a field rotation so that
\begin{align*}
\widetilde T\bigl(\,\underbrace{x-\nabla \phi^\text{inv}(x) - \nabla^\perp \psi^\text{inv}(x)}_\text{\small inverse lensing}\,\bigr) &= T(x).
\end{align*}
Due to the fact that the expected size of the inverse lensing displacement $d(x)$ is on the order of arcmin but the correlation length scale of $\phi$ is on the order of degrees we claim that $-\nabla \phi(f^{-1}(x))$ is well approximated by $-\nabla \phi(x)$. In particular, the divergence-free term $- \nabla^\perp \psi^\text{inv}$ is small  and 
\begin{align}
\label{anti approx}
-\nabla\phi &\approx -\nabla\phi^\text{inv} \approx -\nabla\phi^\text{inv}- \nabla^\perp \psi^\text{inv} = d.
\end{align}
Figure \ref{antilensing plots} illustrates the magnitudes of the above terms. 
The anti-lensing potential $-\phi$  is shown (upper-left) along with the corresponding inverse lensing potential $-\phi^\text{inv}$ (upper-right). The difference $\phi - \phi^\text{inv}$ is also shown (bottom left) along with the stream function $-\psi^\text{inv}$ (bottom-right). Clearly, the magnitude of the difference  $\phi^\text{inv}-\phi$ and $-\psi^\text{inv}$ is sub-dominant to estimation error expected in current lensing experimental conditions.

\begin{figure*}
\begin{center}
  \includegraphics[height=2.4in]{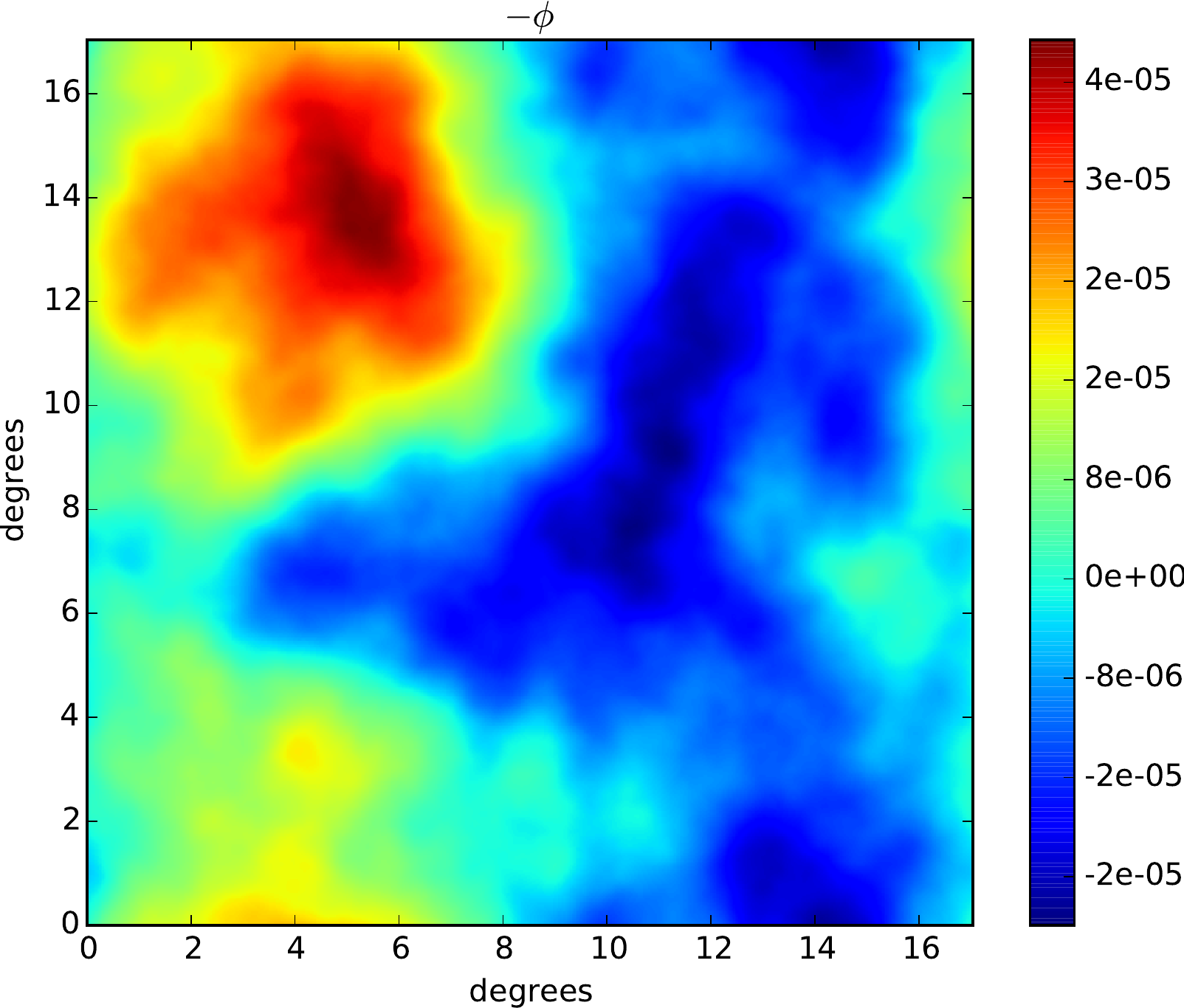}\includegraphics[height=2.4in]{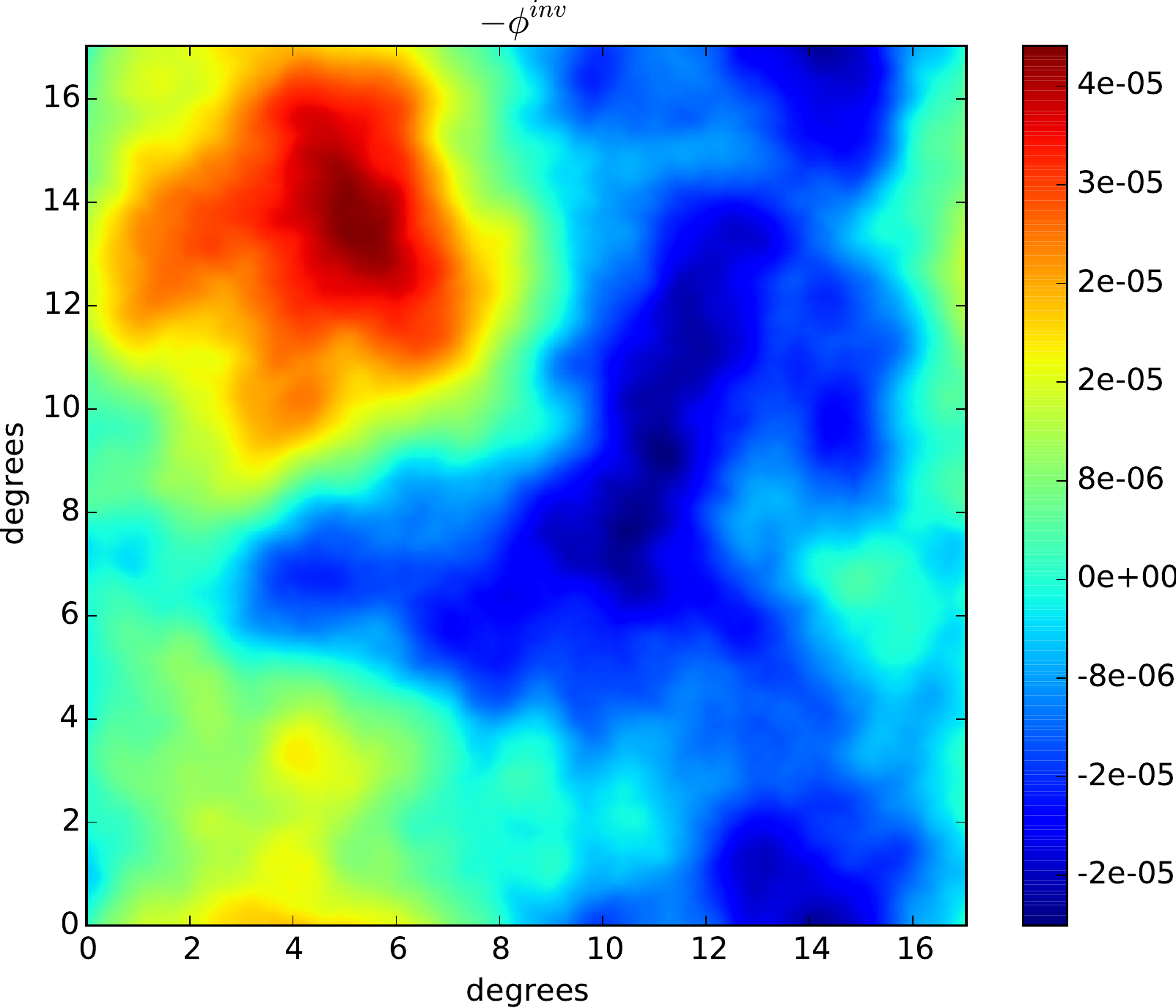}\\%
  \includegraphics[height=2.4in]{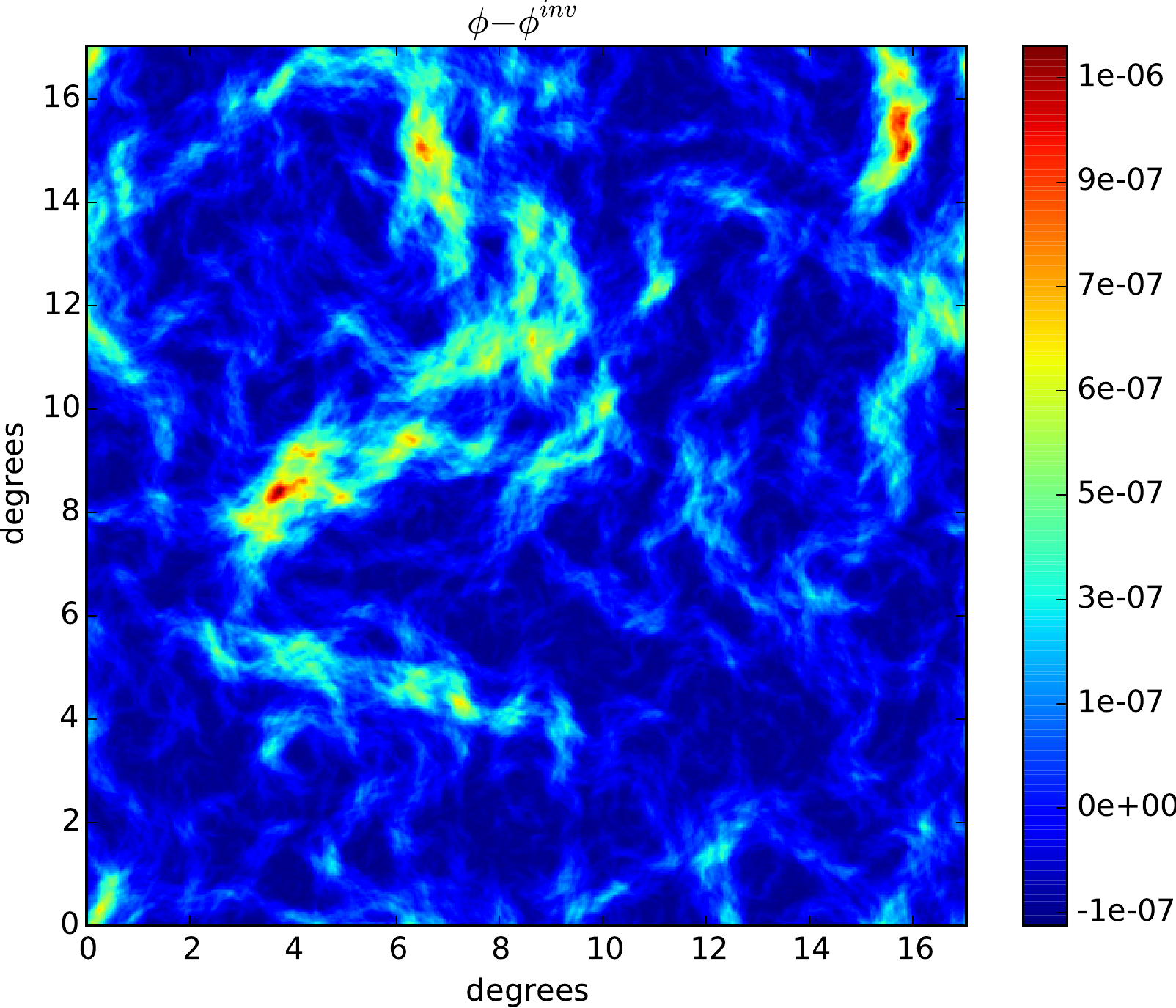}\includegraphics[height=2.4in]{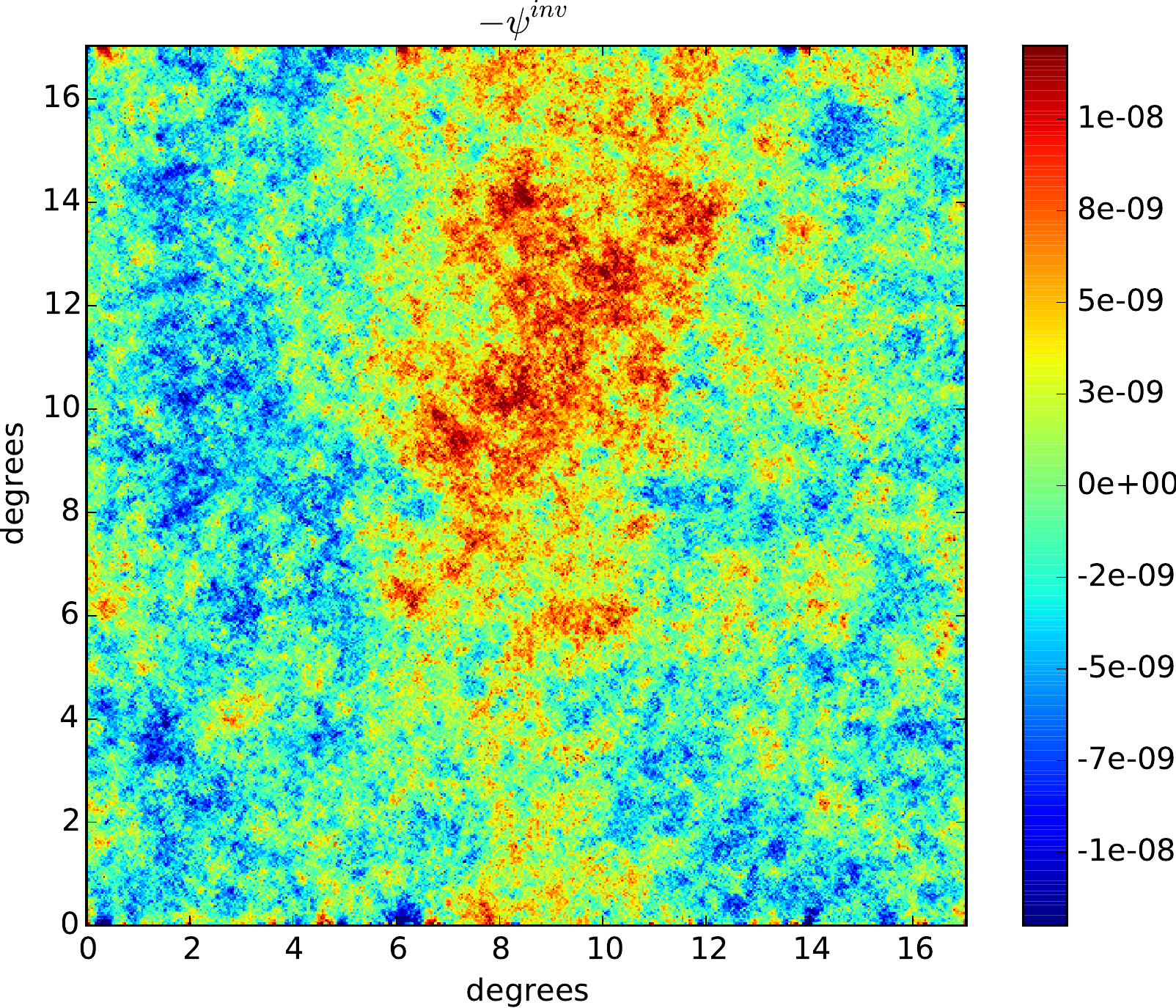}
\end{center}
\caption{\label{antilensing plots}
The difference between anti-lensing and inverse lensing. {\em Upper left:} anti-lensing potential $-\phi$. {\em Upper right:} The inverse lensing potential $-\phi^\text{inv}$. {\em Bottom left:} The difference $\phi^\text{inv}-\phi$. {\em Bottom right:} The inverse lensing stream function $-\psi^\text{inv}$.}
\end{figure*}

%
%
\section{Hamiltonian Monte Carlo sampler for $P(\phi | \widetilde T,  \text{\rm data})$}
\label{Section: hamiltonian sampler section}

The Hamiltonian Monte Carlo (HMC) algorithm is an iterative sampling algorithm designed to mitigate the low-acceptance rate of the Metropolis-Hastings algorithm when working in high dimension. 
A nice review of HMC can be found in \cite{neal2011mcmc}. For applications of HMC in cosmology see \cite{PhysRevD.75.083525, taylor2008fast, elsner2010local, 2010MNRAS.409..355J, 2012MNRAS.425.1042J, jasche2013bayesian, jasche2013methods}. In the present case we utilize the HMC algorithm to produce samples of $\phi$ from $P(\phi | \widetilde T,  \text{\rm data})$.  The key to making HMC work for lensing is to parameterize  $\phi$ in terms of it's Fourier transform. One can then utilize  Claim \ref{grad claim}, presented below, to efficiently compute the gradient of the log conditional density of $P(\phi | \widetilde T,  \text{\rm data})$, which is a necessary computation for the HMC algorithm.

{\em Notation:} Throughout the remainder of this paper, the Fourier transform of any function $f(x)$ will be denoted by $f_l$ or $f_k$ so that $f_l  =  \int_{\Bbb R^2} e^{-i x\cdot l}  f(x)\frac{dx}{2\pi}$ and
$f(x) =  \int_{\Bbb R^2} e^{i x\cdot l}  f_l \frac{dl}{2\pi}$ 
where $l\in \Bbb R^2$ is a two dimensional frequency vector and $x\in \Bbb R^2$  is a two dimensional spatial coordinate.

To describe the HMC algorithm let $\bs \phi$ denote the concatenation of the real and imaginary parts of $\phi_l$ as $l$ ranges through discrete frequencies $l$ ranging up to a pre-specified $|l|_{max}$ (but excluding  half of the Fourier frequencies due to the Hermitian symmetry associated with the Fourier transform of a real field). Note that $\bs \phi$ is a vector of real numbers. 
Let $P(\bs \phi|\widetilde T, \text{data})$ denote the density of $\bs \phi$ given $\widetilde T$ and the data. 
Let $\bs p$ denote a `momentum' vector and $\bs m$ denote a `mass' vector, which are both the same length as $\bs \phi$. The Hamiltonian is a function of $\bs \phi$ and $\bs p$ and is defined as follows
\[ H(\bs \phi, \bs p):= -\log P(\bs \phi|\widetilde T, \text{data})+\sum_k \frac{\bs p_k^2}{2\bs m_k^2}. \]
This Hamiltonian generates a time-dependent evolution of $\bs \phi$ and $\bs p$ given by 
\begin{align*}
\frac{d\bs \phi^t}{dt} &= \phantom{-}\nabla_{\bs p} H(\bs \phi^t, \bs p^t) \\
\frac{d\bs p^t}{dt}    &= -\nabla_{\bs \phi} H(\bs \phi^t, \bs p^t).
\end{align*} 
The HMC is a discrete version of  this time-dynamic equation, using a leapfrog method, which produces a Markov chain $(\bs \phi_1, \bs p_1), (\bs \phi_2, \bs p_2), \ldots$ where the $i^\text{th}$ iteration is given by Algorithm \ref{ith step of HMC} below.

\begin{algorithm}[H]
\small
\caption{ $i^\text{th}$ step of the Hamiltonian Markov Chain}
\label{ith step of HMC}
\begin{algorithmic}[1]
\State Set $\bs \phi^0:= \bs \phi_{i-1}$ and simulate $\bs p^0 \sim \mathcal N(0,\Lambda_m)$ where $\Lambda_m$ is diagonal with  $diag(\Lambda_m)=\bs m$.
\State  Recursively compute $\bs \phi^{k\epsilon}$ and $\bs p^{k\epsilon}$ for $k=1,\ldots, n$ using the following equations:
\begin{align*}
\bs \phi^{t+\epsilon} &:= \bs \phi^{t} + \epsilon \Lambda_m^{-1} \left[ \bs p^t - \frac{\epsilon}{2} \nabla_{\bs \phi} H(\bs \phi^t, \bs p^t) \right], \\
\bs p^{t+\epsilon} &:= \bs p^{t} - \frac{\epsilon}{2}\Bigl[ \nabla_{\bs \phi} H(\bs \phi^t, \bs p^t) + \nabla_{\bs \phi} H(\bs \phi^{t+\epsilon}, \bs p^t)\Bigr].
\end{align*}
\State Simulate $u\sim \mathcal U(0,1)$, and define $p:= \min\left(1,{e^{- H(\bs \phi^{n\epsilon}, \bs p^{n\epsilon})}}/{e^{-H(\bs \phi^0, \bs p^0)}}\right)$.
\State  If $u< p$,  set $\bs \phi_{i}:= \bs \phi^{n\epsilon}$, otherwise set  $\bs \phi_{i}:= \bs \phi_{i-1}$.
\end{algorithmic}
\end{algorithm}

The HMC algorithm is notoriously sensitive to tuning parameters. The prevailing wisdom (see \cite{neal2011mcmc} page 22, for example) is that one should set $\bs m$ to match the reciprocal of the posterior variance of $\bs \phi$. For the simulation presented in Section \ref{Section: simulation examples} we simply set $\bs m^{-1}_l$ to be nearly proportional to $C_l^{\phi\phi}$ with a slight attenuation at low wavenumber. In particular, we set $\bs m_l^{-1} := 2\times 10^2\bigl[\frac{3}{4} +\frac{1}{4}\tanh(\frac{|l| - 1500}{200}) \bigr] C_l^{\phi\phi} \delta_0$ where $\delta_l$ is a discrete dirac in Fourier space.  This choice was motivated by the fact that the high frequency terms $\phi_l$ are not well constrained by the posterior distribution which results in a posterior variance closely matching $C_l^{\phi\phi}$. The remaining parameters of Algorithm \ref{ith step of HMC} are set to $n = 30$ and $\epsilon = 2 \times 10^{-3}u$ where $u$ is a uniform $(0,1)$ random variable sampled anew at each pass of Algorithm \ref{ith step of HMC} (the use of random $\epsilon$ is designed to avoid resonant frequencies, as advocated in \cite{taylor2008fast}).

The key difficulty in using Algorithm \ref{ith step of HMC} is the computation of the $\nabla_{\bs \phi} H(\bs \phi^t, \bs p^t)$, or equivalently the computation of $\nabla_{\bs \phi} \log P(\bs \phi|\widetilde T, \text{data})$. The number of frequencies is extremely large and therefore, any slow computation of the gradients will present a serious bottleneck. 
The follow claim shows that the gradient of the log density of $P(\phi|\widetilde T, \text{data})$, with respect to the Fourier basis of $\phi$, can be computed quickly with Fourier and inverse Fourier transforms. 

\begin{claim}
\label{grad claim}
 Under the anti-lensing approximation (\ref{antilense}) for any nonzero frequency vector $l\equiv (l_1, l_2) \in \Bbb R^2$ 
\begin{equation}
\label{grad claim eq}
 \frac{\partial}{\partial \phi_l}\log P(\phi | \widetilde T,  \text{\rm data}) \propto -  \frac{\phi_l}{C^{\phi\phi}_{l}} -  \sum_{q=1,2} i  l_q \int_{\Bbb R^2} e^{-i x\cdot l} A^q(x)B(x) \frac{dx}{2\pi}  
\end{equation}
 where  $\phi_l = \re \phi_l + i \im \phi_l$, $\frac{\partial}{\partial \phi_l}\equiv \frac{\partial}{\partial\re \phi_l} + i \frac{\partial}{\partial\im\phi_l}$ and 
 \begin{align}
 B_l &\equiv \frac{1}{C_l^{TT}} \int e^{-i x\cdot l}  \widetilde T(x-\nabla \phi(x))\frac{dx}{2\pi} \\ 
 A^q(x) &\equiv \frac{\partial\widetilde T}{\partial x_q}\bigl(x-\nabla \phi(x)\bigr).
 \end{align}
\end{claim}

An important fact used in the derivation of (\ref{grad claim eq}) is that the lensing and anti-lensing operator is invertible. For example, if $\phi(x)$ and $\widetilde T(x)$ are known at all pixel locations $x$, then it is possible to perfectly reconstruct $T(x)$. This implies that the anti-lensing operation (which is a linear action on the CMB) can be represented as an infinitesimal permutation matrix. Therefore, the determinant of the anti-lensing operator $\det( d \widetilde T^\phi / d\widetilde T)$ equals $1$, where $\widetilde T^\phi(x)\equiv \widetilde T(x-\nabla \phi(x))$.   
 Now, to compute the likelihood surface $P(\phi | \widetilde T,  \text{\rm data})$ as a function of $\phi$ we obtain the following formula:
\begin{align*}
P(\phi | \widetilde T,  \text{\rm data})  
&= P(\phi | \widetilde T)  \\
&\propto  P(\widetilde T | \phi)P(\phi) \\
&=  \underbrace{|\det( d \widetilde T^\phi / d\widetilde T)|}_{=1} P(\widetilde T^{\phi} | \phi)P(\phi)
\end{align*}
where $P(\widetilde T^\phi|\phi)$ represents the likelihood that $\widetilde T^{\phi}$ is statistically unlensed by $\phi$. In other words, $P(\widetilde T^\phi|\phi)$ measures the likelihood that $\widetilde T^{\phi}$ is an  isotropic Gaussian random field with spectral density $C_l^{TT}$. This explains the following characterization of the log likelihood of $\phi$ given $\widetilde T$ and the data:
\begin{align}
\log P(\phi | \widetilde T,  \text{\rm data}) =  c - \frac{1}{2}\int_{\Bbb R^2}\left[ \frac{\bigl|\widetilde T_k^\phi\bigr|^2}{C_k^{TT}} + \frac{|\phi_k|^2}{C_k^{\phi\phi}} \right] dk 
\end{align}
where $c$ is a constant which does not depend on $\phi$.  The remaining details of the derivation of Claim \ref{grad claim} is left to the appendix.

It is instructive to compare the gradient calculation (\ref{grad claim eq}) with the quadratic estimate of $\phi$ 
 developed in \cite{hu2001mapping,hu2002mass}. The quadratic estimate, applied to observations of the form $ \widetilde T(x) + n(x)$, is given by 
 \begin{equation}
\label{quad estimate}
 \hat \phi_l = - N_l  \sum_{q=1,2} i  l_q \int_{\Bbb R^2} e^{-i x\cdot l} A^q(x)B(x) \frac{dx}{2\pi}  
\end{equation}
where  $B_l  \equiv \bigl[C_l^{\widetilde T\widetilde T}+ C_l^{nn}\bigr]^{-1}\bigl[\widetilde T_l + n_l \bigr]$, $A^q_l\equiv il_q\bigl[ C_l^{TT}\bigr]\bigl[C_l^{\widetilde T\widetilde T} + C_l^{nn}\bigr]^{-1} \bigl[\widetilde T_l + n_l \bigr]$ and 
\[
N_l^{-1} \equiv \frac{1}{2} \int_{\Bbb R^2}  \frac{\bigl(l\cdot(k+l) C^{TT}_{k+l} - l\cdot k C^{TT}_{k} \bigr)^2}{\bigl(C^{\widetilde T\widetilde T}_{k+l}+ C_{k+l}^{nn} \bigr)\bigl(  C^{\widetilde T\widetilde T}_{k}+ C_{k}^{nn} \bigr)}\frac{dk}{(2\pi)^2}.
\] 
The term $N_l$ is radially symmetric in frequency $l$ and corresponds to a normalization which makes the quadratic estimate unbiased up to first order. After substituting $ C_l^{\widetilde T\widetilde T} + C_l^{nn}\rightarrow C_l^{TT}$ and $n_l \rightarrow 0$ in the formula for the quadratic estimate, one obtains 
\[
 \frac{\partial}{\partial \phi_l}\log P(\phi | \widetilde T,  \text{\rm data})\bigr|_{\phi = 0} = \frac{\hat \phi_l}{N_l}.
\]
Indeed, an approximate Newton step, using (\ref{grad claim eq}), is an accurate approximation to the quadratic estimate $\hat \phi_l$.  This illustrates how the parameterization $(\widetilde T, \phi)$ results in Gibbs iterations which make drastic moves, on the order of the size of the quadratic estimate. 

One of the features of the quadratic estimate is that the fast Fourier transform (FFT) and inverse fast Fourier transform (IFFT) can be used to compute $\hat \phi_l$ for all frequencies $l$. Naively computing the quadratic form of the quadratic estimate requires $O(n^2)$ flops, rather than the $O(n\log n)$ flops obtained by the FFT/IFFT method, where $n$ denotes the number of pixels. We note that Claim \ref{grad claim} establishes that the gradient computation inherits a similar FFT/IFFT characterization to compute $ \frac{\partial}{\partial \phi_l}\log P(\phi | \widetilde T,  \text{\rm data})$ at all frequencies $l$, in $O(n\log n)$ flops. Since this gradient computation needs to be embedded in a  Hamiltonian Markov step within a Gibbs Chain, the computation efficiency gained by the FFT/IFFT is absolutely crucial.

%
%
\section{Iterative message passing algorithm for $ P(\widetilde T |  \phi,\text{\rm data})$}
\label{Section: iterative message passing section}

There are two natural ways to model the lensed CMB $\widetilde T$. If one marginalizes out $\phi$, then $\widetilde T$ is modeled as a {\em non-Gaussian} but isotropic random field. Conversely, if one conditions on $\phi$ the field $\widetilde T$ is modeled as a {\em non-isotropic} but Gaussian random field. The latter case is relevant for sampling from $ P(\widetilde T |  \phi,\text{\rm data})$ which  is, therefore,  simply a Gaussian conditional simulation problem. Unfortunately, the non-isotropic (indeed, non-stationary) nature of the conditional distribution of $\widetilde T$ presents serious computational challenges. In what follows we utilize a new iterative algorithm developed in \cite{elsner2013efficient} for Gaussian conditional expectation when the signal is diagonalized in harmonic space that the noise is diagonalized in pixel space. The method we present here is similar to the Gibbs sampling adaptation of \cite{jasche2014matrix}.

Start by transforming each pixel location $x$ by the lensing operation $x+\nabla \phi(x)$, while simultaneously preserving the data associated with that pixel. This effectively de-lenses  $\text{data}(x)= T(x+\nabla \phi(x))+n(x)$ but produces observations on an irregular grid. In particular, one may switch to the lensed coordinates $y = x+\nabla\phi(x)$ so that
\begin{align}
\nonumber
\underbrace{(x +\nabla \phi(x),\, \text{ data}(x))}_{\text{(pixel, data) tuple}} & = (y, \, T(y) + \tilde n(y)) 
\end{align}
where $\tilde n(x+\nabla\phi(x))=n(x)$.  Now the data $(y, \, T(y) + \tilde n(y))$ is arranged on an irregular grid in $y$. This irregular grid is then embedded into a high resolution regular grid by nearest neighbor interpolation. The points $y$ which do not get assigned an observation $T(y)+\tilde n(y)$ under the interpolation we consider to be masked. 
Figure \ref{embed} illustrates this situation. The left hand plot shows the irregularly sampled data $(x +\nabla \phi(x),\, \text{ data}(x))$ and the right hand plot shows the grid embedding. The filled dots represent observations of $T(y) + \tilde n(y)$ whereas the empty dots correspond to a masked observation of $T(y)$.  Finally we extend the definition of $\tilde n(y)$ to have infinite variance over the masked region, whereby producing data $T(y)+\tilde n(y)$ over a dense regular grid in $y$.

\begin{figure*}
\begin{center}
{\includegraphics[height=2.5in]{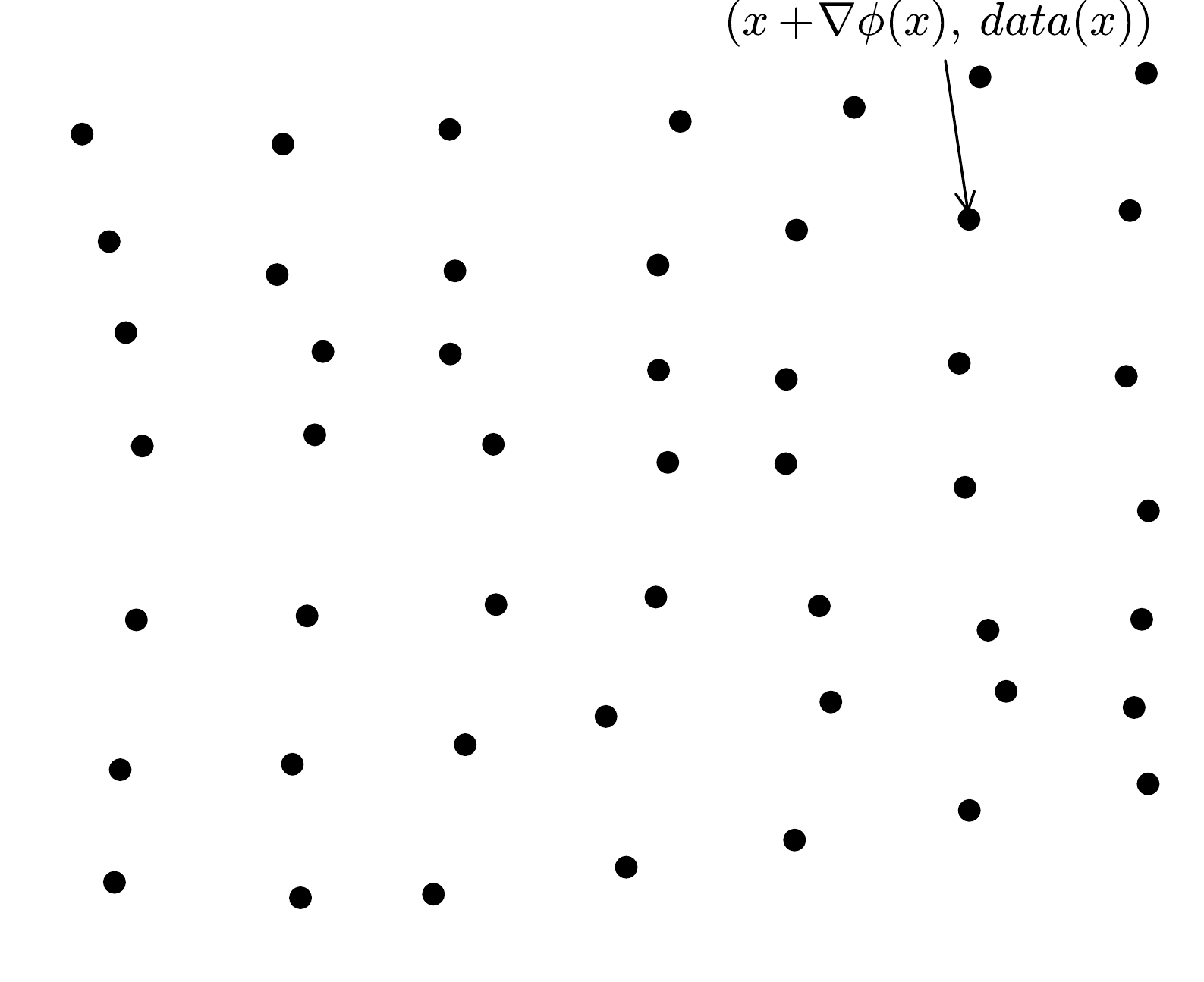}}%
{\includegraphics[height=2.6in]{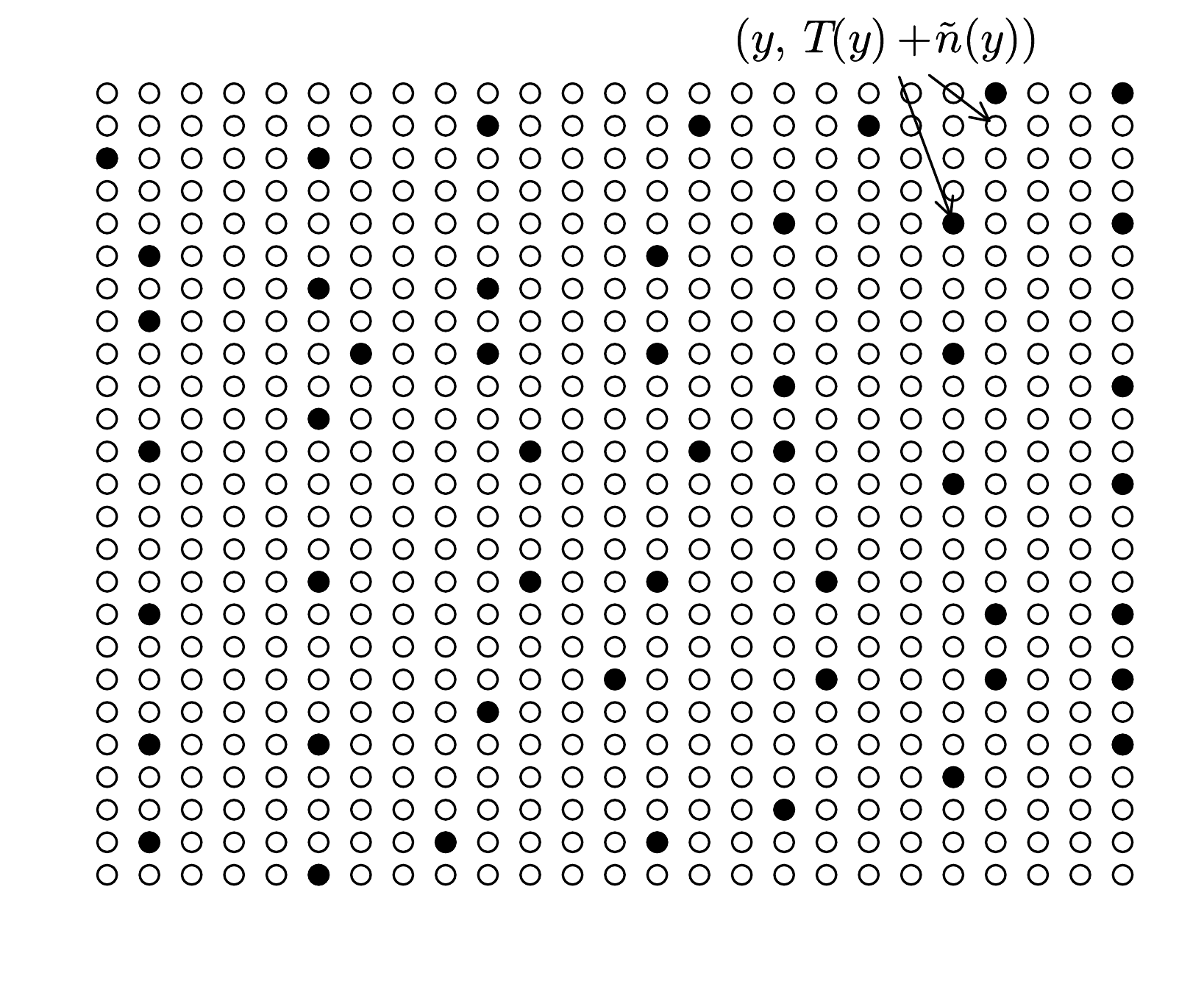}}
\end{center}
 \caption{
\label{embed}
 This graphic illustrates how knowledge of $\phi(x)$, used when sampling from the Gibbs step $P(\widetilde T|\phi,\text{data})$, converts white noise corrupted gridded observations of the lensed CMB into masked observations of the unlensed CMB on a more dense grid.  The data associated with the original grid, indexed by $x$, is moved via advection to the lensed grid $x+\nabla \phi(x)$ seen at left. The right panel shows the lensed grid embedded into a higher resolution grid. The data, indexed by the dense regular grid on the right panel, is of the form $T(y)+\tilde n(y)$ where $T(y)$ denotes the unlensed CMB and $\tilde n(y)$ denotes the noise. The unobserved locations, represented by open circles, are characterized with an infinite variance for  $\tilde n(y)$. 
} 
\end{figure*}

As a intermediate step in producing a sample from  $P(\widetilde T |  \phi,\text{\rm data})$ we produce a conditional  sample of $T(y)$ given the observations $T(y) + \tilde n(y)$. The difficulty of this step is that $\tilde n(y)$ is non-homogeneous noise---from the masking and any inhomogeneity in $n(x)$---and therefore it is not decorrelated by the Fourier transform. To handle this situation we adapted a new method for Gaussian conditional expectation developed in \cite{elsner2013efficient}. This method works particularly well for observations with large amounts of irregular masking, as in our case. The algorithm utilizes a messenger field which effectively behaves as a latent---signal plus white noise---model which is amenable to Gibbs sampling \citep{jasche2014matrix}.

The delensing algorithm described in this paper requires the capability to do fast constrained realization of non-lensed CMB. The embedding illustrated in Figure \ref{embed} means that each constrained realization needs to be computed on a mask with a great deal of structure on the scale of the pixels of dense grid. Several algorithms exist in the literature to solve  the general problem of constrained Gaussian random field on a given mask and power spectrum: the conjugate gradient method \citep{wandelt2004cg,eriksen2004cg}, the multiscale conjugate gradient method \citep{smith2007mcg}, the multigrid method \citep{seljebotn2014mg}, the Messenger algorithm \citep{elsner2013efficient} and its variant the Gibbs-Messenger \citep{jasche2014matrix}. 

Since the lensing potential changes from each iteration to the next, every constrained realization of non-lensed CMB needs to be computed for a different set of active points in the embedding grid. This effectively means that the solution is done for a different mask at every iteration.  This rules out linear solvers that require expensive pre-computations, e.g. of pre-conditioners, that depend on the coefficient matrix of the system since that depends on the mask. We also require  exact acceptance and higher speed than direct or standard iterative methods. This reduces the possibilities to either the Messenger algorithm or the Gibbs-Messenger.  

The Gibbs-Messenger generates very fast constrained realizations that converge to the correct distribution in a statistical sense without iterating to a numerical solution (thus obviating the need to specify a colling schedule in Algorithm~\ref{IMP}) for the price of losing independence between subsequent samples. In contrast the Messenger algorithm simulates independent constrained realisations, but requires iteration of the linear system. To ensure a numerically accurate solution implies a conservative choice of cooling schedule  in Algorithm~\ref{IMP} (though this is still much faster than the alternatives described in the previous paragraph). 

We adopt a hybrid approach where we occasionally generate an independent sample  using the Messenger algorithm, and then generate many quick samples using the Gibbs-Messenger approach. The detailed choices for the cooling schedule and the number of samples between  full Messenger solutions will be described in Section \ref{Section: simulation examples}.

\begin{algorithm}[H]
\small
\caption{Algorithm for sampling from $P(T | \text{data})$ where $\text{data}(y) =  T(y) + \tilde n(y)$}
\label{IMP}
\begin{algorithmic}[1]
\State Set cooling schedule $\lambda_1\geq  \ldots \geq \lambda_n$ where $\lambda_n = 1$. 
\State Decompose $\tilde n(y)$ into a homogeneous part with variance
$\bar\sigma^2$ and a non-homogeneous part with variance $\tilde\sigma^{2}(y)$ so that 
\begin{quote}
$\text{var}(\tilde n(y)) = \bar\sigma^2 + \tilde\sigma^{2}(y)$ 
\end{quote}
where $\tilde \sigma^2(y) = \infty$ on all masked pixels $y$. Notice that the spectral density of the homogeneous part is given by $ \bar\sigma^2 dy$ where  $dy$ denotes the pixel grid area. 
\State Initialize the fields $M(y)$ and $T(y)$  to be zero at all pixel locations $y$.
\State Recursively update fields $M(y)$ and $T(y)$ by iterating the following steps for $j = 1, \ldots, n$:
\begin{quote}
\begin{itemize}
\item[$\bullet$] Simulate a mean zero Gaussian random field $Z(y)$ which is independent across pixels and with pointwise variance $\bigl(\frac{1}{\lambda_j\bar\sigma^{2}} + \frac{1}{\tilde\sigma^{2}(y)}\bigr)^{-1}$.
\item[$\bullet$] Update 
$M(y) \leftarrow \text{data}(y)\displaystyle\frac{\lambda_j\bar\sigma^{2}}{\lambda_j \bar\sigma^2 + \tilde \sigma^2(y)} + T(y)\frac{\tilde\sigma^{2}(y)}{\lambda_j\bar\sigma^2 + \tilde \sigma^2(y)}  + Z(y) $.
\item[$\bullet$] Simulate a mean zero Gaussian random field, $W(y)$, with spectral density
$\langle W_l W_{l^\prime}^*\rangle = \delta_{l-l^\prime}\bigl( \frac{1}{C^{TT}_l} + \frac{1}{\lambda_j \bar\sigma^2 dy} \bigr)^{-1}$
\item[$\bullet$] Update  $T_l \leftarrow  M_l\displaystyle\frac{C^{TT}_l}{C^{TT}_l + \lambda_j   \bar\sigma^2 dy} + W_l$. 
\end{itemize}
\end{quote}
\State Return $T(x)$.
\end{algorithmic}
\end{algorithm}

The following algorithm describes the use of Algorithm~\ref{IMP} to produce a sample from $ P(\widetilde T |  \phi,\text{\rm data})$.

\begin{algorithm}[H]
\small
\caption{Sampling from $P(\widetilde T |  \phi,\text{\rm data})$}
\label{after HMC}
\begin{algorithmic}[1]
\State Embedded the  pixel/data pairs  $(x, \text{data}(x))$ into observations of the form $(y, T(y)+\tilde n(y))$ where $y$ ranges over a high resolution regular grid as illustrated in the right plot of Figure \ref{embed}.
\State Use Algorithm \ref{IMP} to produce a sample $T\sim P(T\,|\, T+\tilde n)$.
\State Return $\widetilde T(x) = T(x+\nabla\phi(x))$.
\end{algorithmic}
\end{algorithm}

At present, algorithms \ref{IMP} and \ref{after HMC} are designed for the situation that the pixels are sufficiently small compared to the magnitude of $\nabla \phi(x)$ and the noise is approximately white on these scales. 
Indeed, our goal is to explore the low noise and small beam experimental conditions where the quadratic estimate is known to be suboptimal (see \cite{HirataSeljak1, HirataSeljak2}). 
That being said, the only change needed to incorporate other experimental details in the Bayesian lensing methodology presented here, including foreground contaminants, is how one samples from $P(\widetilde T |  \phi,\text{\rm data})$. Algorithms \ref{IMP} and \ref{after HMC} take advantage of the special lensed-grid structure of the data when conditioning on $\phi$ to accomplish this goal. Adding different/new experimental details to the data will still result in a Gaussian constrained realization problem. We acknowledged that more complicated modeling of the data will most certainly introduce additional computational challenges. However, we consider these computational challenges to be sub-dominant to the fundamental bottleneck for Bayesian lensing which was the extremely slow mixing time of the original Gibbs formulation. Moreover, the structure of the new Gibbs formulation isolates all experimental details to the Gibbs step (\ref{suff 2}) where conditioning on the lensing potential $\phi$ results in a classic Gaussian constrained realization problem.

%
%
\section{Simulation example}
\label{Section: simulation examples}

In this section we present a simulation to illustrate the methodology presented above. The simulated lensing potential used in this section, shown at left in Figure \ref{phix fig}, is generated on a flat sky with periodic boundary conditions. The data, shown upper-left in Figure \ref{tilde fig}, is generated on 2 arcmin pixels with independent additive noise and masking. The noise level is set to $8.0$ $\mu K$ arcmin and the masking covers approximately 10\% of the pixels.  The parameters of the Bayesian lensing procedure are the Fourier modes of $\widetilde T$ and $\phi$. For the lensing potential we set $|l|_\text{max}$ to $~460$. 
For this sky coverage the scale-resolution in Fourier space  $\Delta l = 21$ yields $1500$ unknown Fourier coefficients for $\phi_l$. The $|l|_\text{max}$ of $2700$ for the unlensed temperature $T$ is set in Algorithm \ref{IMP} and corresponds to half of the Nyquist limit at $2$ arcmin pixels.

\begin{figure}
\begin{center}
  {\includegraphics[width=.5\hsize]{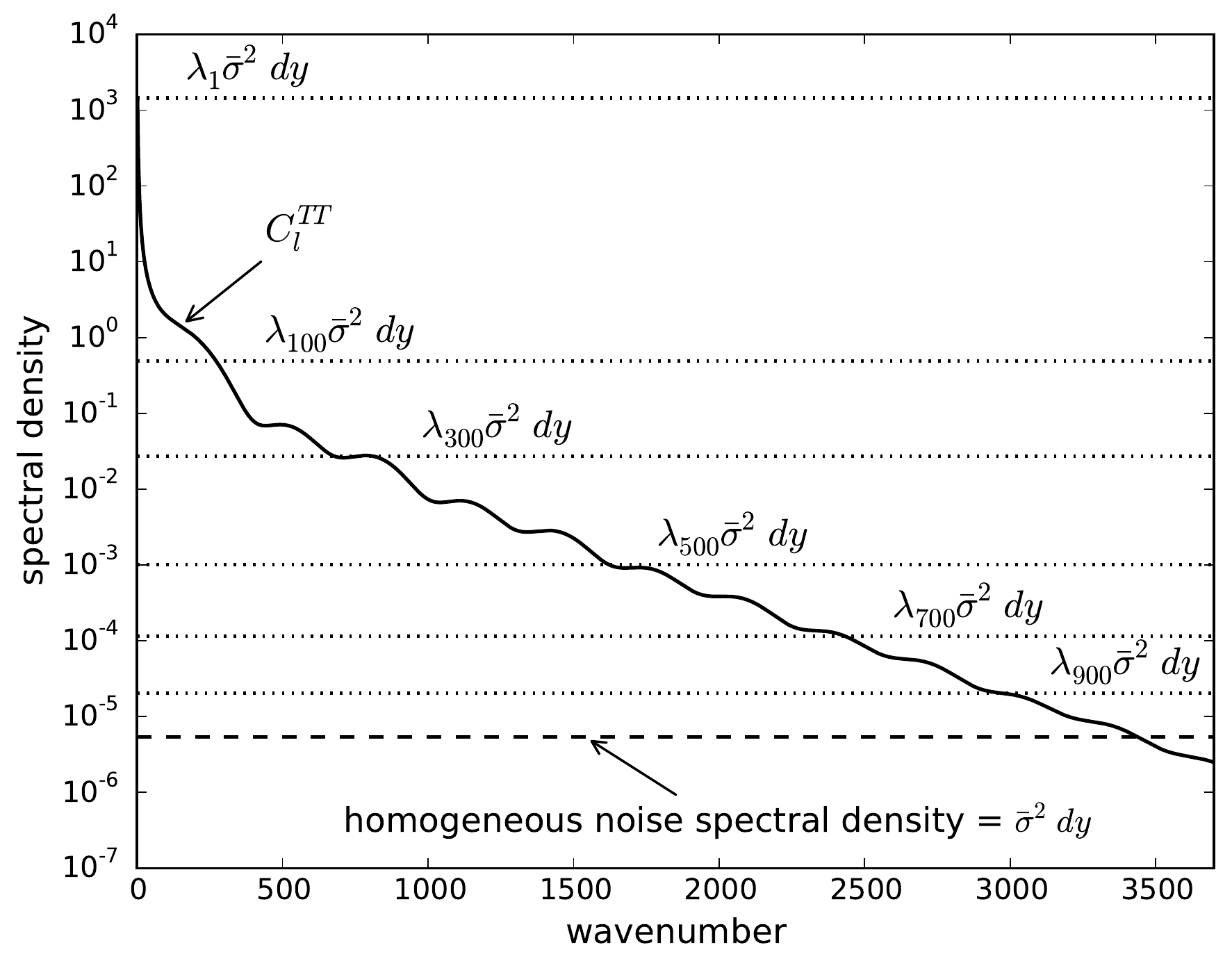}}
\end{center}
\caption{
\label{cooling}
This figure shows values of the cooling parameter $\lambda_j$ for $j=1,100, 300, 500, 700, 900$ used in Algorithm \ref{IMP}.  This schedule is applied every $100^\text{th}$ step in the Gibbs algorithm.  In Algorithm $\ref{IMP}$, the value of $\lambda_j\bar{\sigma}^2dy$ serves as the spectral density of artificial additive white noise in the latent field $M(x)$. Therefore setting $\lambda_j$ greater than $1$,  encourages fast mixing of $T_l$ at all frequency vectors $l$ such that $C_l^{TT} \lesssim \lambda_j\bar{\sigma}^2dy$. The cooling schedule shown above is an attempt to let $\lambda_j$ approach $1$ in such a way as to encourage all frequency vectors up to $|l|_\text{max}$ to mix quickly. 
 }
\end{figure}

We ran $10$ parallel Gibbs chains for a total of $2500$ steps. The timings for each Gibbs iteration averaged approximately 200 seconds using a Dual Intel Xeon E5-2690 2.90GHz processor. Each chain was initially warmed up by replacing the HMC draws in the first $5$ iterations with a gradient ascent. A burn-in of approximately $550$ runs were discarded and the remaining runs were thinned by $100$. The result is a total of $200$ posterior samples. The cooling schedule for the iterative message passing algorithm was selected by numerical experimentation.  Most of the Gibbs iterations set the cooling terms $(\lambda_1,\ldots, \lambda_{400}) \equiv (1,\ldots, 1)$  in Algorithm \ref{IMP}. However, we did find it advantageous to periodically run a nontrivial $1000$-step cooling schedule every $100^\text{th}$ pass of the Gibbs algorithm. This nontrivial cooling schedule for $\lambda_j$ is plotted in Figure \ref{cooling} and is set in an attempt to encourage fast mixing of the Fourier modes $T_l$ up to $|l|_\text{max}$.

The best Bayesian estimate of $\phi(x)$ corresponds to the posterior mean $E(\phi(x) |\text{data})$. This quantity is approximated by the average of the $200$ draws from the Gibbs chain and is shown in the middle plot of Figure \ref{phix fig}.  The right plot of Figure \ref{phix fig} shows the quadratic estimate of $\phi(x)$ for comparison. However, due to the difficulty when using the quadratic estimate in the presence of sky cuts, the quadratic estimate shown uses all of the data---including the pixels which are masked---in producing the estimate of $\phi$. In general, one can see good agreement with $E(\phi(x) | \text{data})$  and $\phi(x)$. Indeed, the effect of masking is visually undetectable as compared to the quadratic estimate.  To get a better visualization of the individual draws from the posterior, the left plot in Figure \ref{slice fig} shows a horizontal cross section of the posterior draws of $\phi(x)$ taken at vertical degree mark $12.7^o$.

The Gibbs methodology presented here yields samples of the lensed CMB $\widetilde T(x)$ and the lensing potential $\phi(x)$ conditional on the data. Moreover, as a byproduct of Algorithm \ref{after HMC}, we also obtain samples of the unlensed $T(x)$ given the data. By averaging $200$ draws from the Gibbs chain one can construct an approximation to $E(T(x) | \text{data})$, shown in the upper-right plot of Figure \ref{tilde fig}. In the bottom-left plot of Figure \ref{tilde fig} we show the  difference $T(x) - E(T(x) | \text{data})$.
When compared to the nominal difference between lensed and unlensed CMB $T(x) - \widetilde T(x)$, shown bottom-right in Figure \ref{tilde fig}, one can see that the Gibbs methodology is successful at delensing the observed CMB. To get a better visualization of the individual draws from the posterior, the right plot in Figure \ref{slice fig} shows a horizontal cross section of the posterior draws of $T(x)$ taken at vertical degree mark $12.7^o$ and is magnified near the masking region for better visual inspection.

\begin{figure*}
\begin{center}
{\includegraphics[height=2.1in]{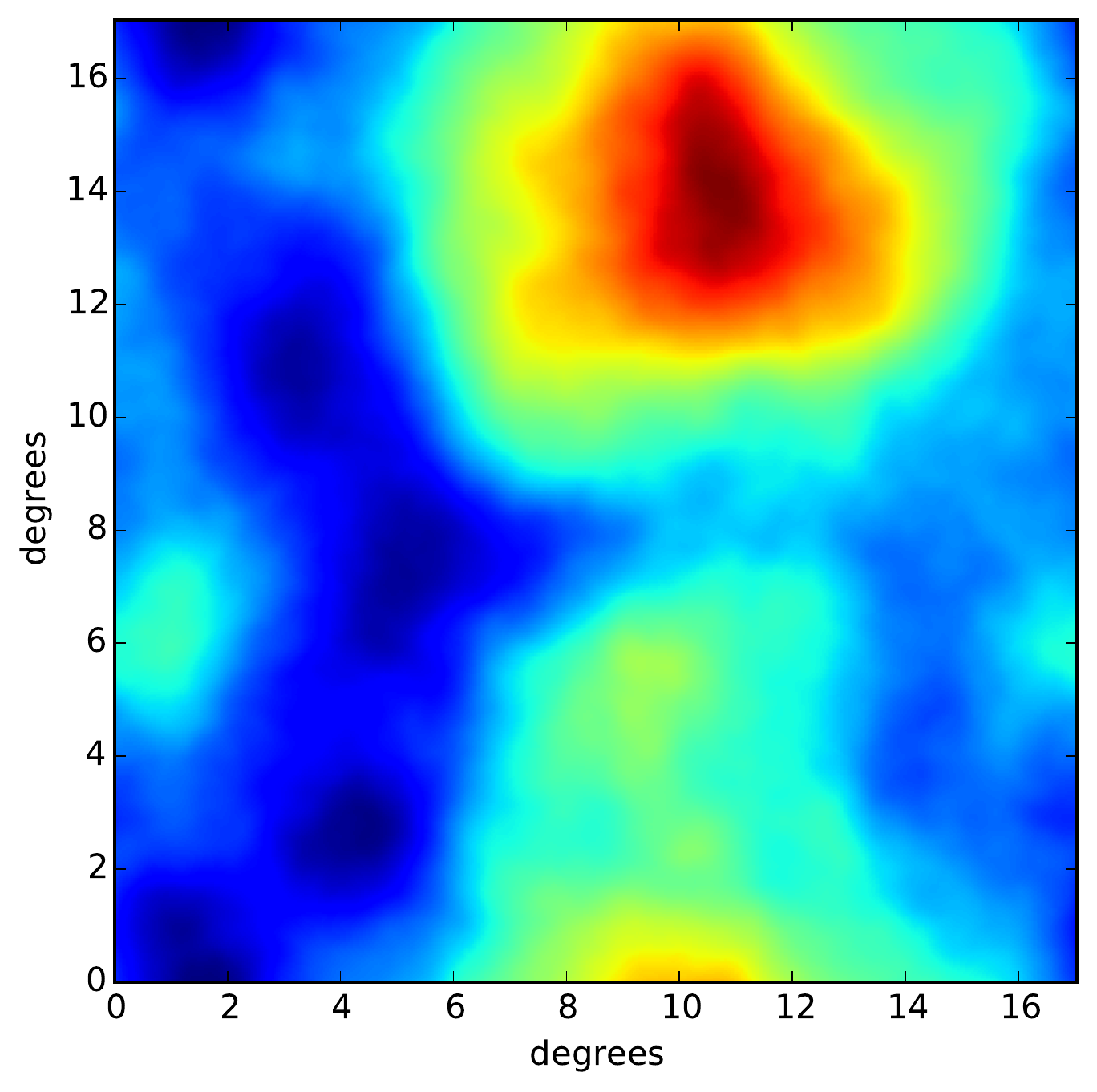}}%
{\includegraphics[height=2.1in]{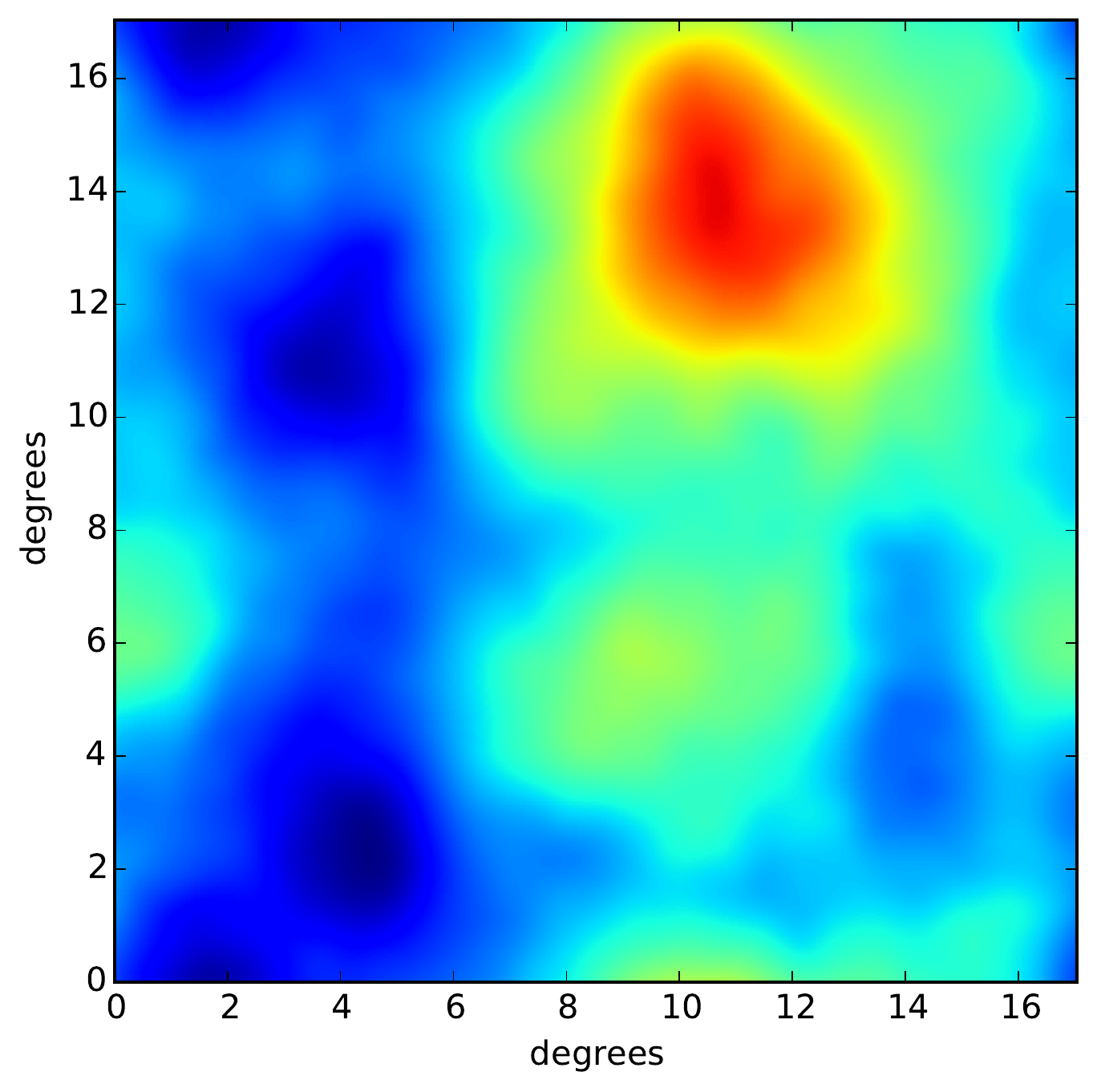}}%
{\includegraphics[height=2.1in]{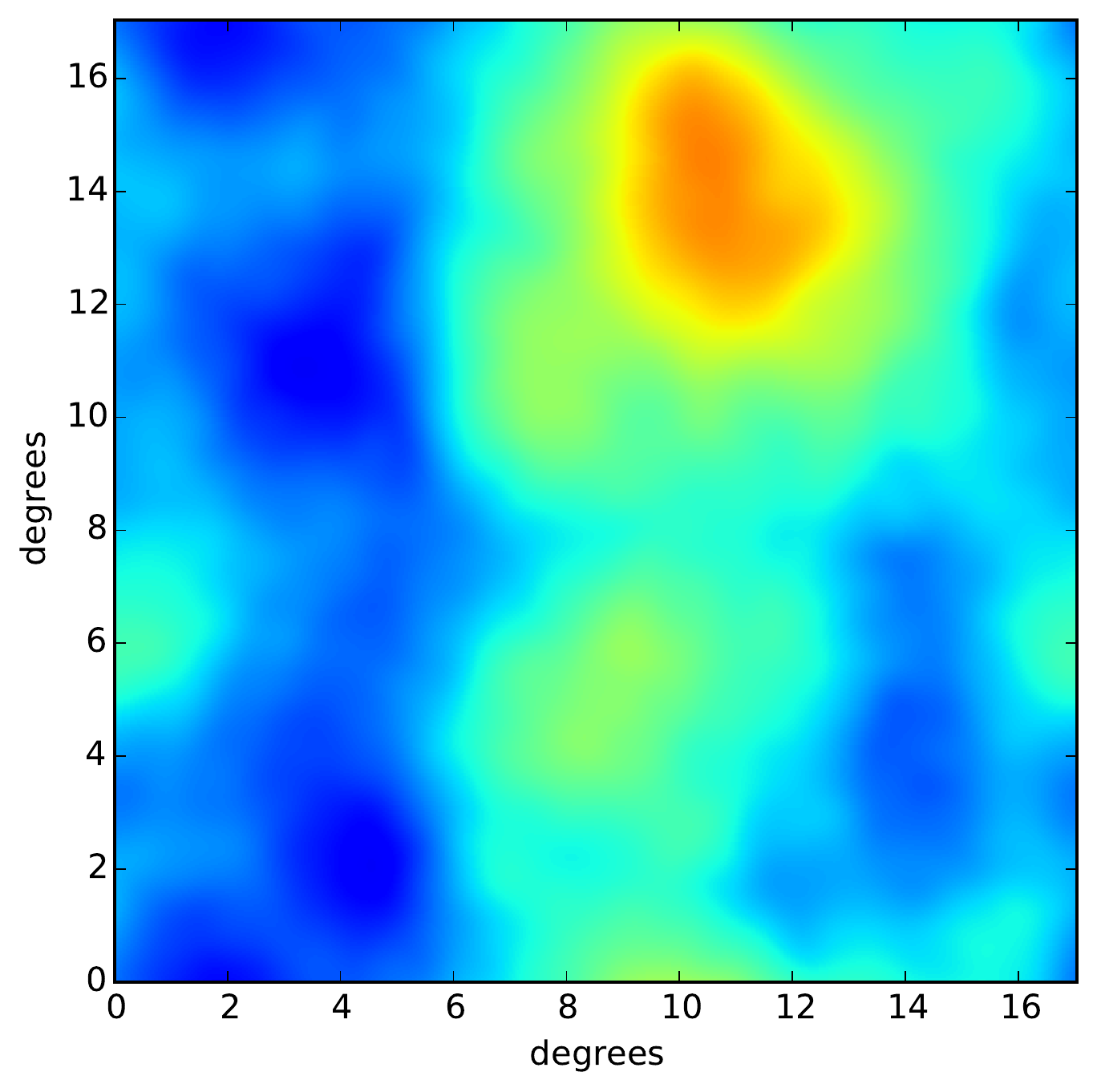}}%
\end{center}
\caption{\label{phix fig}
{\em Left:} simulation truth $\phi(x)$. {\em Middle:} posterior mean $E(\phi(x)|\text{data})$. {\em Right:} The quadratic estimate. To avoid difficulties associated with masked data when using the quadratic estimate,  the estimate shown at right is applied to the full data set with the masked region removed. In contrast, the data used for the Bayesian methodology is masked as shown in Figure \ref{tilde fig}.}
\end{figure*}

\begin{figure*}
\begin{center}
{\includegraphics[height=2.6in]{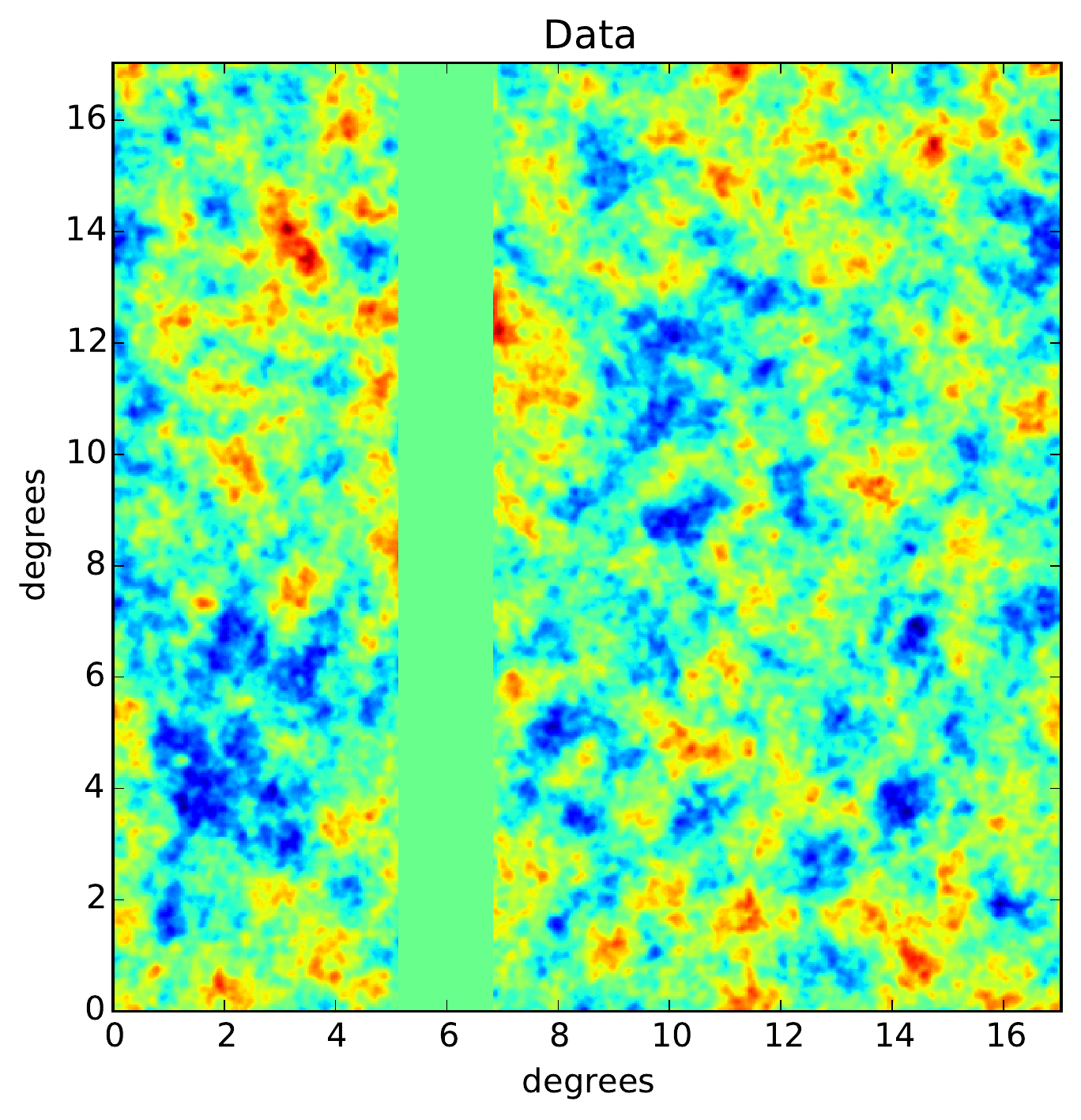}}%
{\includegraphics[height=2.6in]{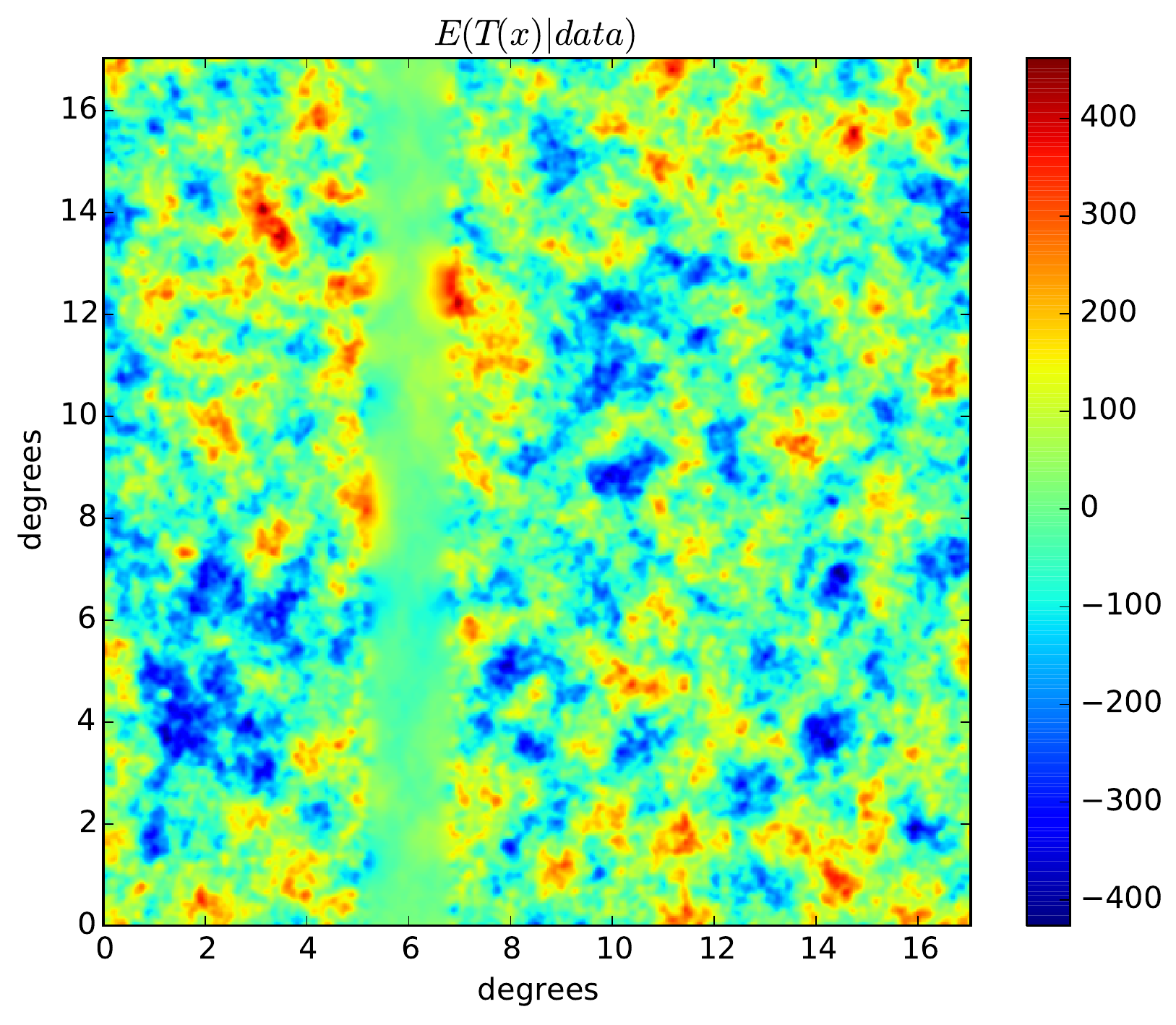}}\\
{\includegraphics[height=2.6in]{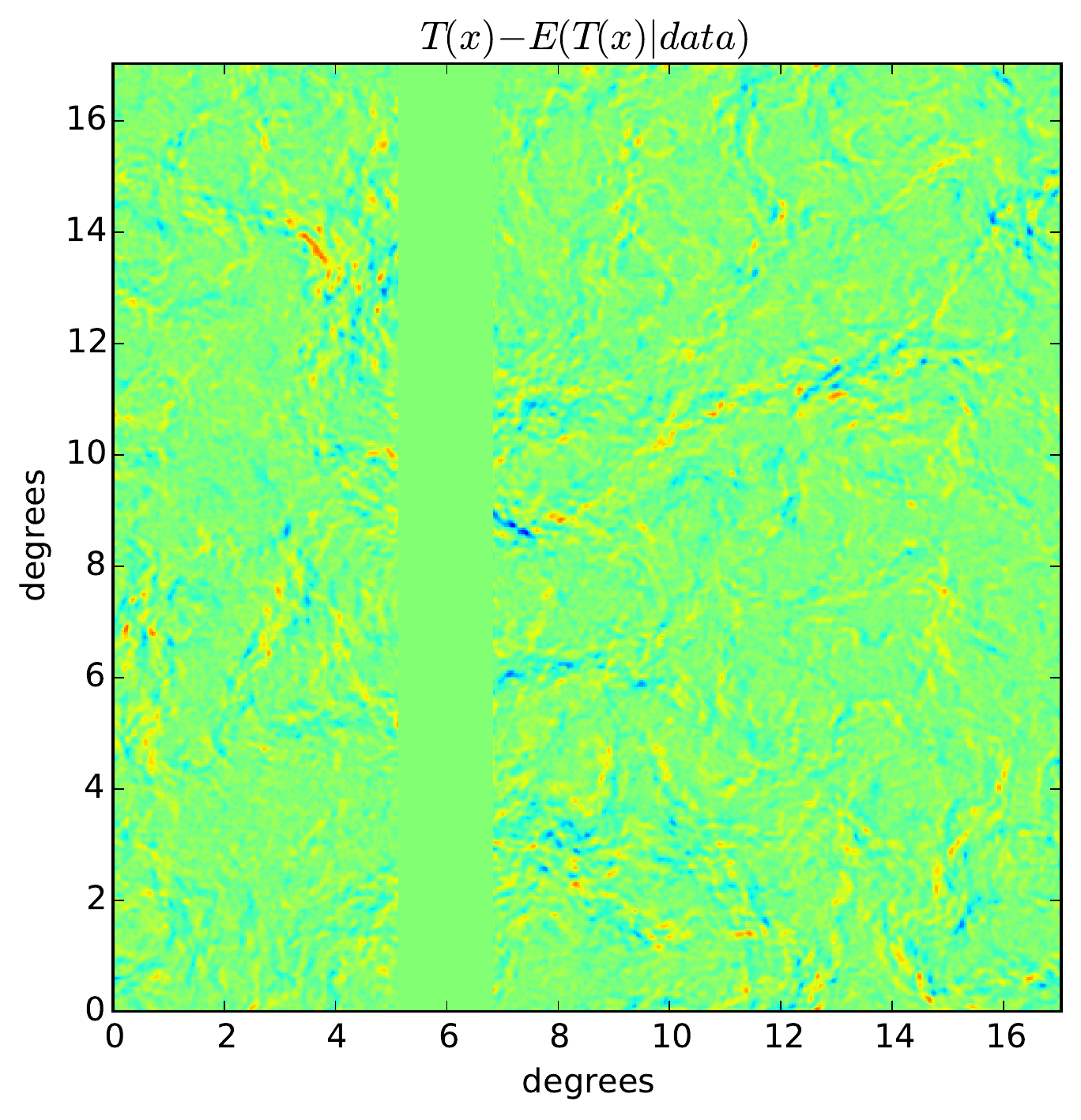}}%
{\includegraphics[height=2.6in]{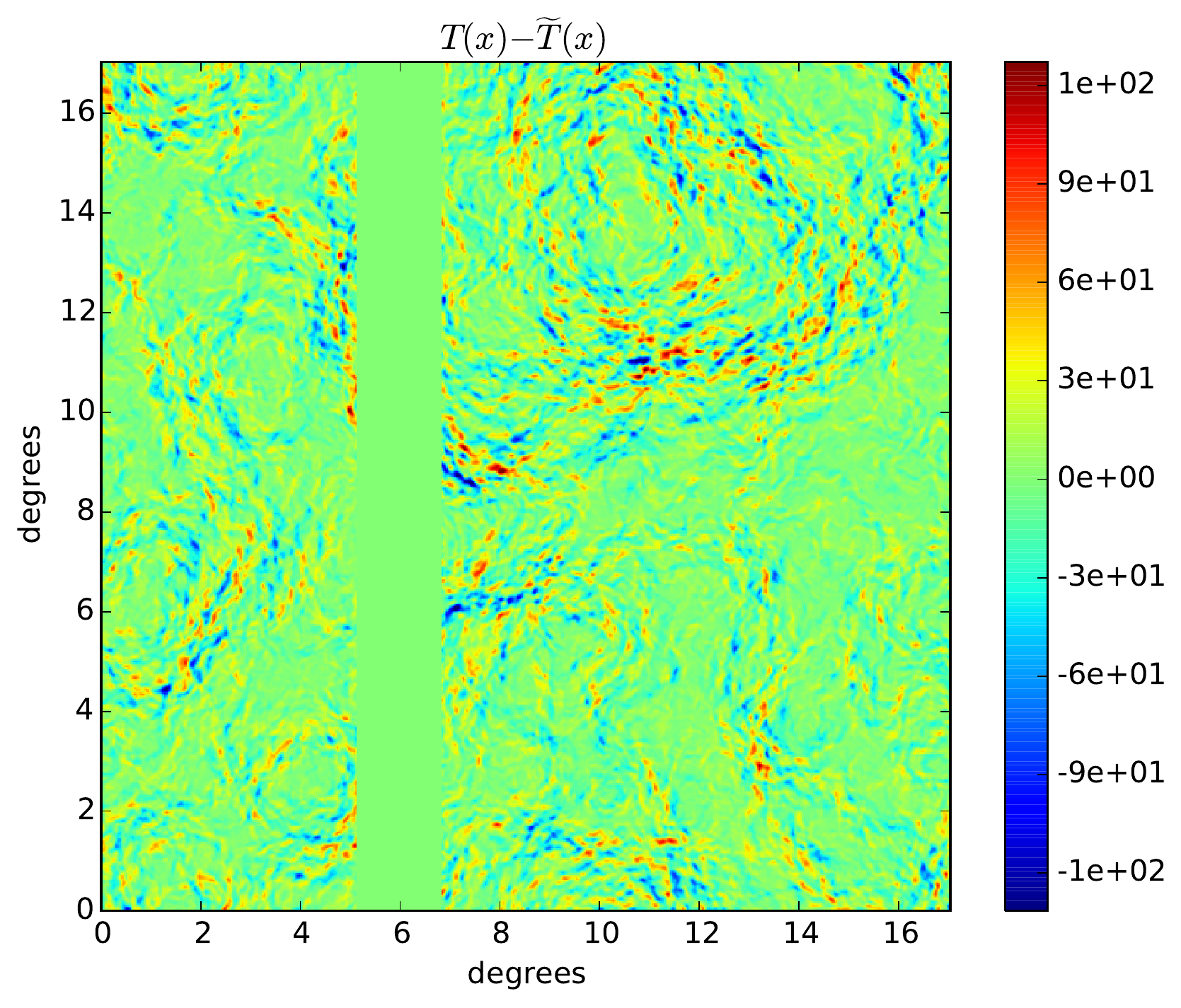}}%
\end{center}
\caption{\label{tilde fig}
{\em Upper left:} simulated lensed CMB data with masking and additive white noise (at level $8.0$ $\mu K$ arcmin). {\em Upper right:} posterior mean $E(T(x)|\text{data})$. {\em Lower left:} This plot shows $T(x) - E(T(x)|\text{data})$ and probes the ability of the Bayesian methodology to delense the observations. This plot should be compared with the nominal difference between the simulation truth unlensed CMB and the lensed CMB,  $T(x) - \widetilde T(x)$, shown {\em bottom right}.
 }
\end{figure*}

\begin{figure*}
\begin{center}
{\includegraphics[height=2.5in]{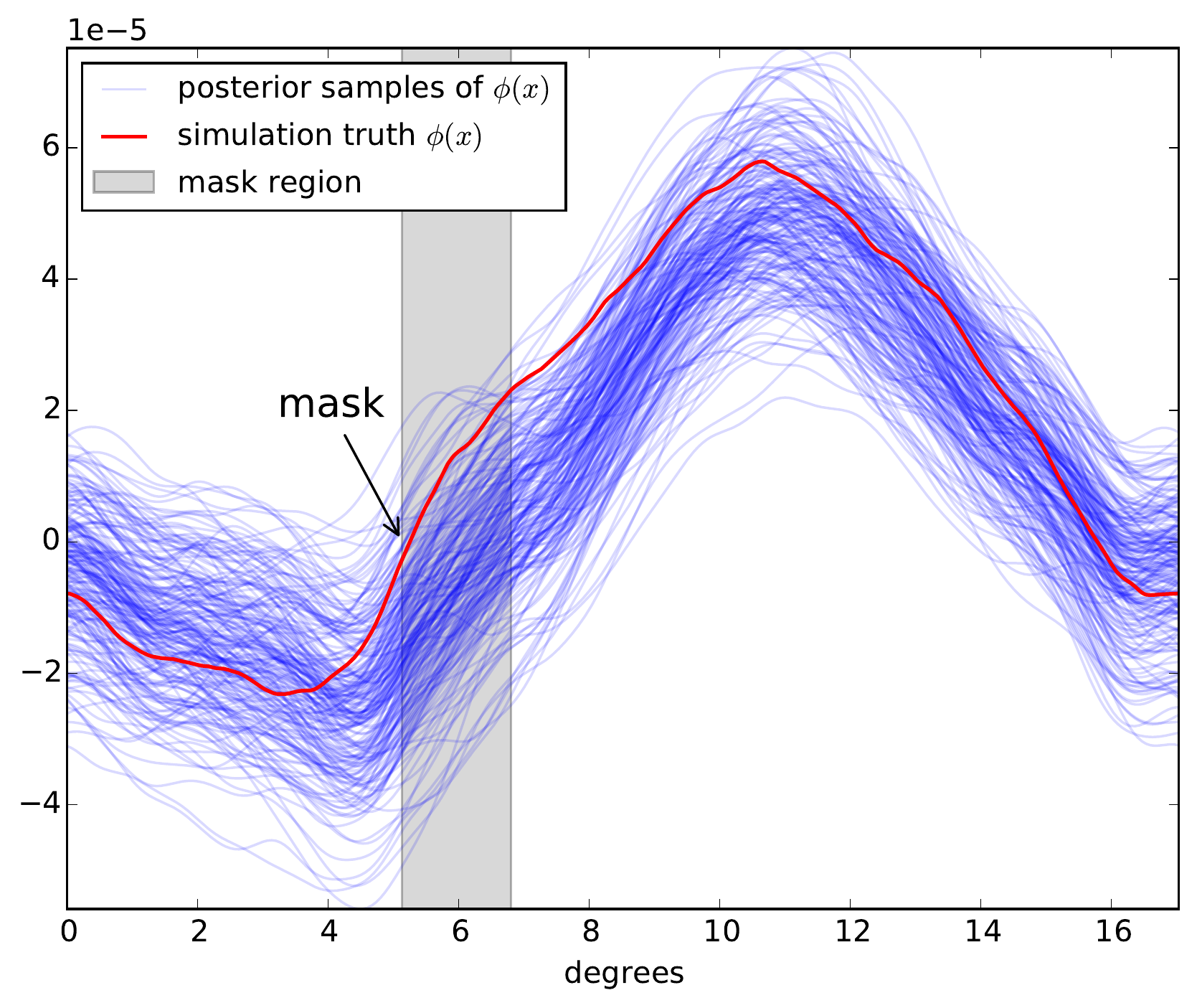}}%
{\includegraphics[height=2.5in]{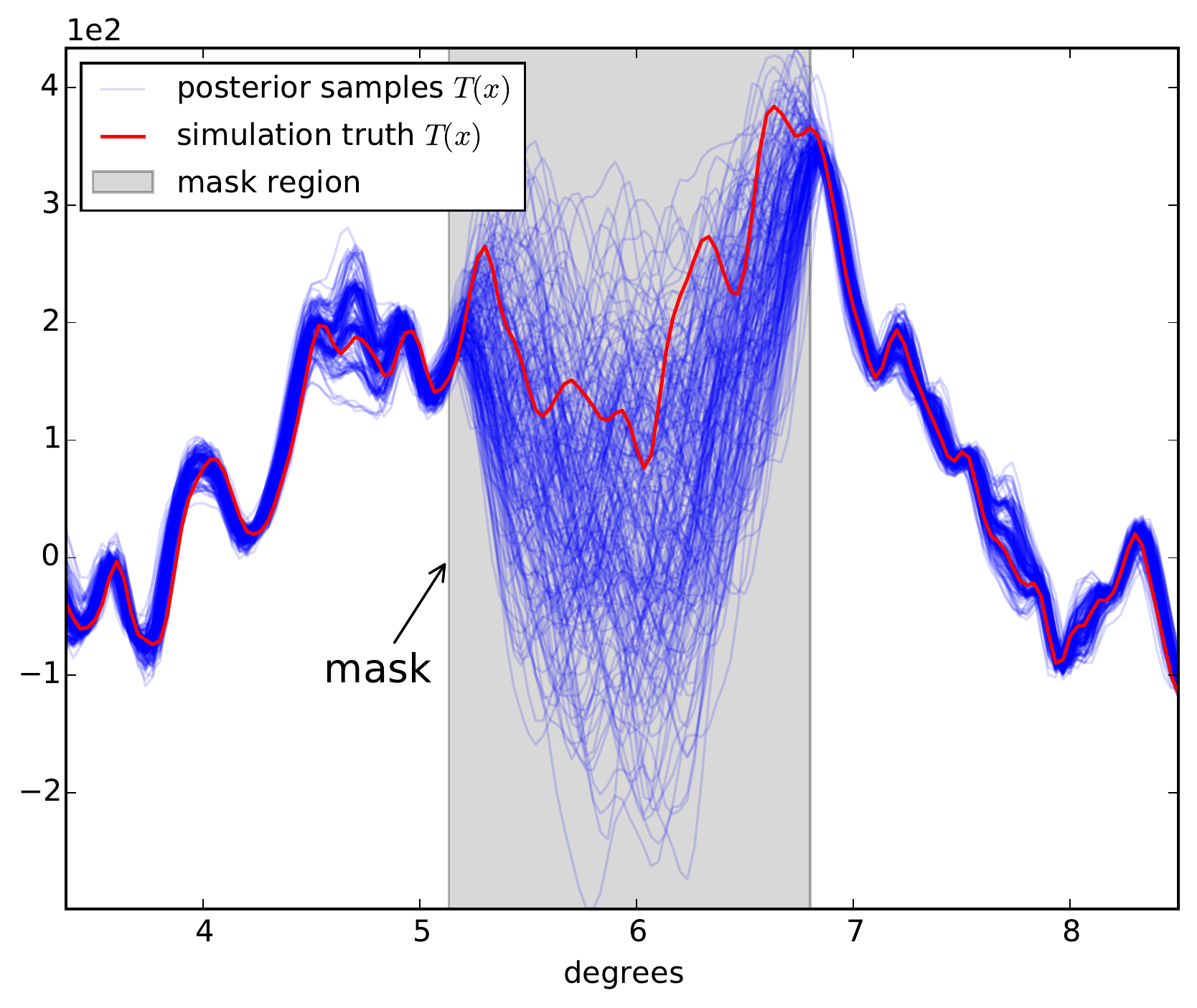}}
\end{center}
\caption{\label{slice fig} Here we plot one dimensional slices of the posterior draws from $P(\phi(x)|\text{data})$ and $P(T(x)|\text{data})$. For reference, these slices are taken from the horizontal regions at vertical degree mark $12.7^o$ in Figures \ref{tilde fig} and \ref{phix fig}.
}
\end{figure*}

In Figures \ref{p spec fig} and \ref{t spec fig} we summarize the posterior draws for $\phi$ and $T$ in the Fourier domain. The left plot of Figure \ref{p spec fig} shows $95\%$ posterior regions for  $l^4|\phi_l|^2/(4\delta_0)$  averaged over $l$ in wavenumber bins. For comparison the simulation true values of  $l^4|\phi_l|^2/(4\delta_0)$ are shown in red and the spectral density $l^4 C^{\phi\phi}_l/4$ is plotted in the black solid line. 
The same quantities are shown for $l^2|T_l|^2/(\delta_0)$ in the right plot of Figure \ref{p spec fig}. Finally, in Figure \ref{t spec fig} we show the empirical cross correlation of the posterior draws over wavenumber bins between the simulation truth and the posterior samples of $\phi_l$ and $T_l$. In particular,
samples were generated from $\frac{1}{c}\sum_{l\in \Delta l} \phi^\text{sim}_l\phi^*_l $ and $\frac{1}{c}\sum_{l\in \Delta l} T^\text{sim}_lT^*_l $ where $\Delta l$ is a frequency wavenumber bin,  $\phi_l$ and $T_l$ are the simulation truth, $\phi^\text{sim}_l$  and $T^\text{sim}_l$  are sampled from the Gibbs algorithm presented here and $c$ is a normalization constant which transforms to a correlation scale.  Notice that the plotted correlations trend to $0$ for larger wavenumber. This is what one would expect since larger wavenumber have correspondingly less information which causes the posterior to revert back to the prior.

In Figure \ref{assessing convergence fig}  we show the Gibbs chain correlation length scale  and the speed of mixing for different statistics of the lensing potential. The left plot shows the Gibbs chain for $\phi(x)$ where $x = (9.9^o, 13.2^o)$ and $x = (1.6^o, 1.6^o)$ in the same degree coordinates given in Figures \ref{tilde fig} and \ref{phix fig}. The right plot shows the real and imaginary parts of $\phi_l$ where the frequency vector $l$ is set to  $(126.56,63.28)$. Each dashed line represents the corresponding simulation truth parameters. These plots suggest that the Gibbs chain is mixing well and that the correlation length scale is small enough so that thinning by $100$ is sufficient to yield relatively uncorrelated samples.

\begin{figure*}
\begin{center}
{\includegraphics[height=2.5in]{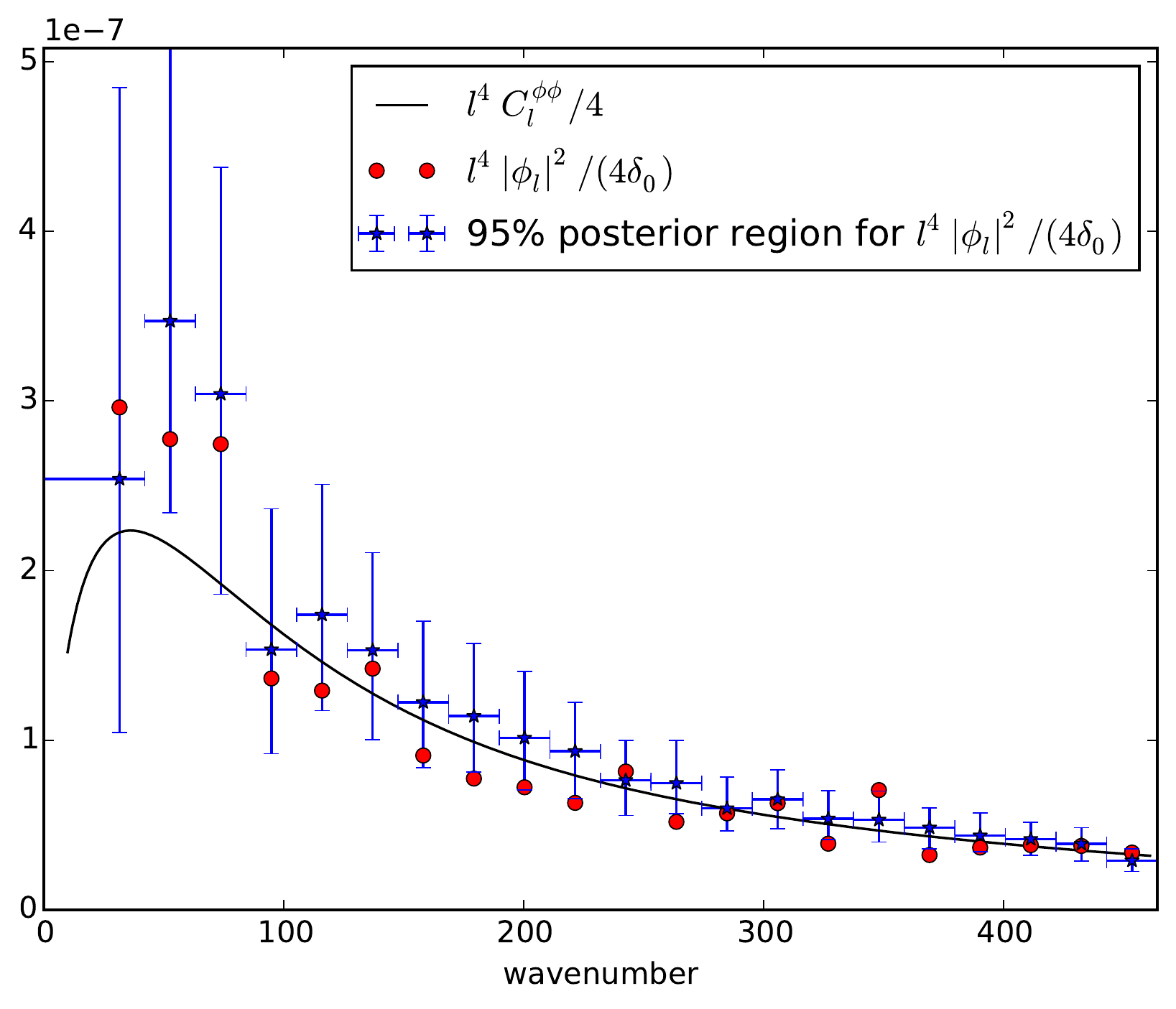}}%
{\includegraphics[height=2.5in]{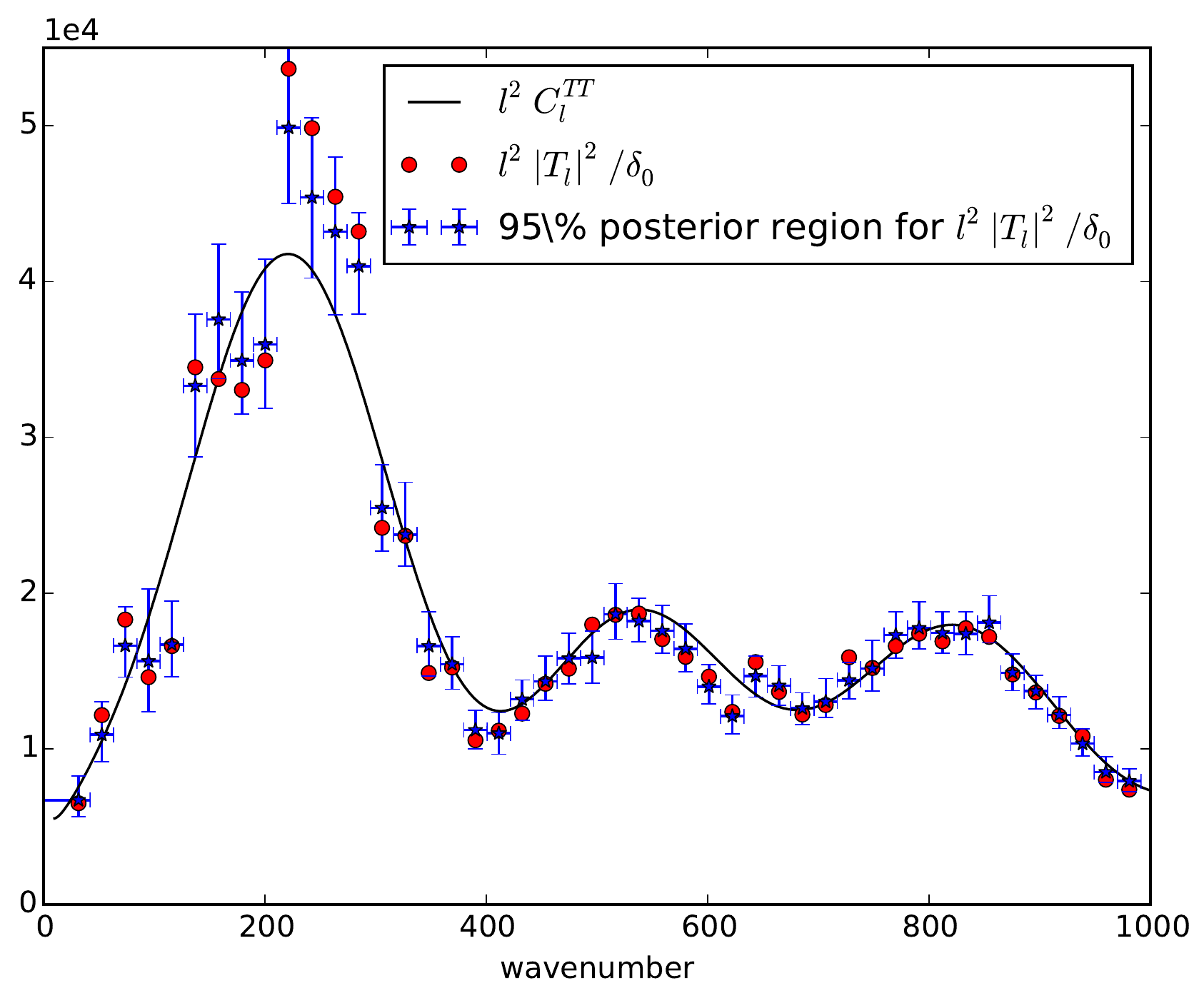}}
\end{center}
\caption{\label{p spec fig}
Estimates of $l^4|\phi_l|^2/4$  and $l^2|T_l|^2$ (shown in blue), scaled to the units of the corresponding spectral density. The red dots show $l^4|\phi_l|^2/4$  and $l^2|T_l|^2$ for the simulation truth (similarly scaled). The discrepancy between the red dots and the spectral densities, shown in black, is exclusively due to cosmic variance. The confidence bars show $95\%$ probability regions from the posterior distributions $P(l^4|\phi_l|^2/4\,|\,\text{data})$ and $P(l^2|T_l|^2\,|\,\text{data})$.
}
\end{figure*}

\begin{figure*}
\begin{center}
{\includegraphics[height=2.5in]{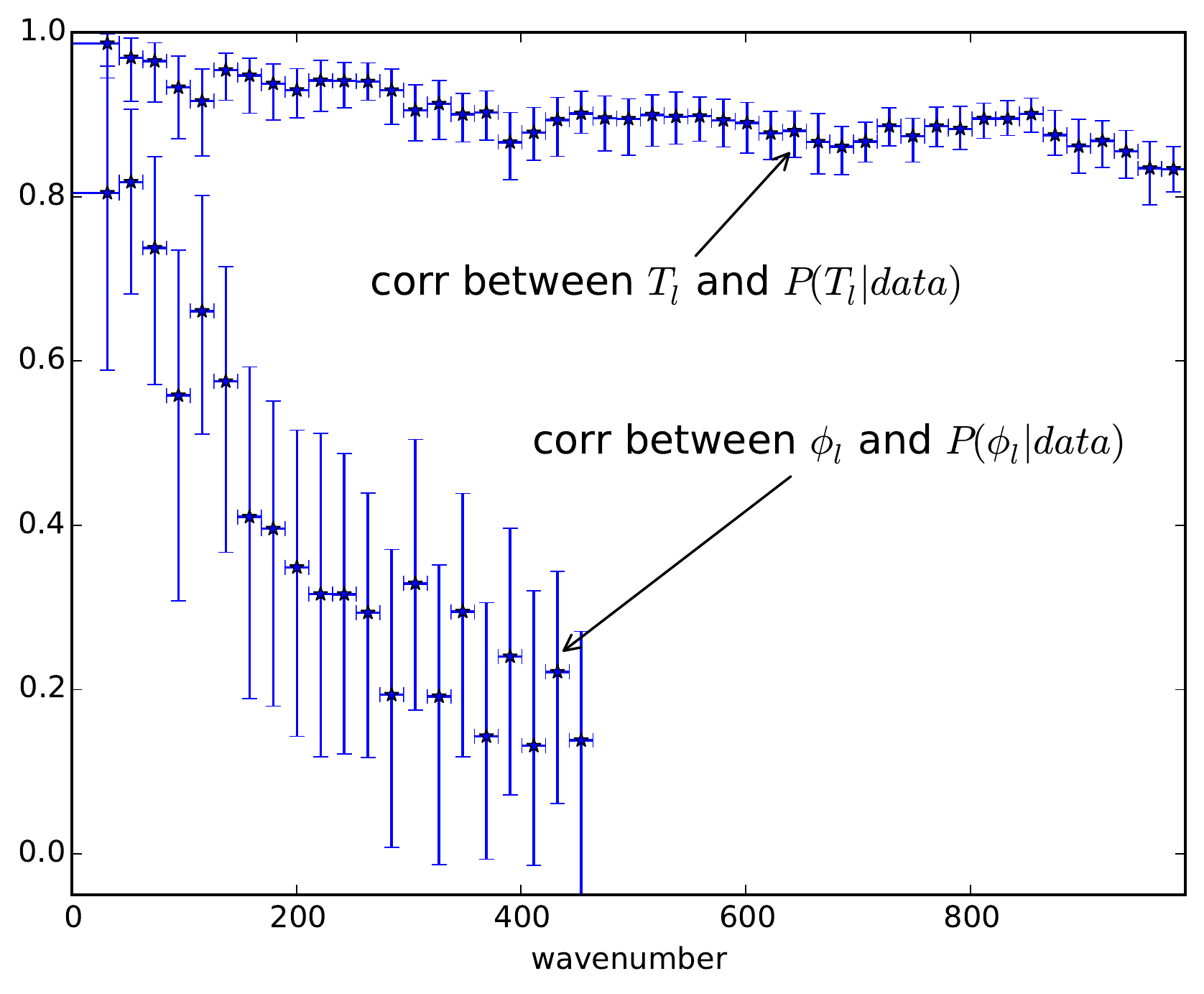}}%
\end{center}
\caption{\label{t spec fig}
This plot summaries the correlation, in $\Delta l$ wavenumber bins, between the simulation truth and their corresponding posterior samples. In particular,
samples were generated from $\frac{1}{c}\sum_{l\in \Delta l} \phi^\text{sim}_l\phi^*_l $ and $\frac{1}{c}\sum_{l\in \Delta l} T^\text{sim}_lT^*_l $ where $\Delta l$ is a frequency wavenumber bin,  $\phi_l$ and $T_l$ are the simulation truth, $\phi^\text{sim}_l$  and $T^\text{sim}_l$  is sampled from the Gibbs algorithm presented here and $c$ is a normalization constant which transforms to a correlation scale.  Recall that $|l|_\text{max}$ for the lensing potential is $\sim460$ and is $\sim2700$ for the unlensed CMB, which explains why the correlation for $\phi_l$ only extends to $460$.
Notice that the plotted correlations trend to $0$ for larger wavenumber.
This is what one would expect from the Bayesian posterior. Indeed, the data is less informative at larger wavenumber. This causes the Bayesian posterior to revert to the prior, which will be uncorrelated with the simulation truth. 
}
\end{figure*}

\begin{figure*}
\begin{center}
{\includegraphics[height=2.5in]{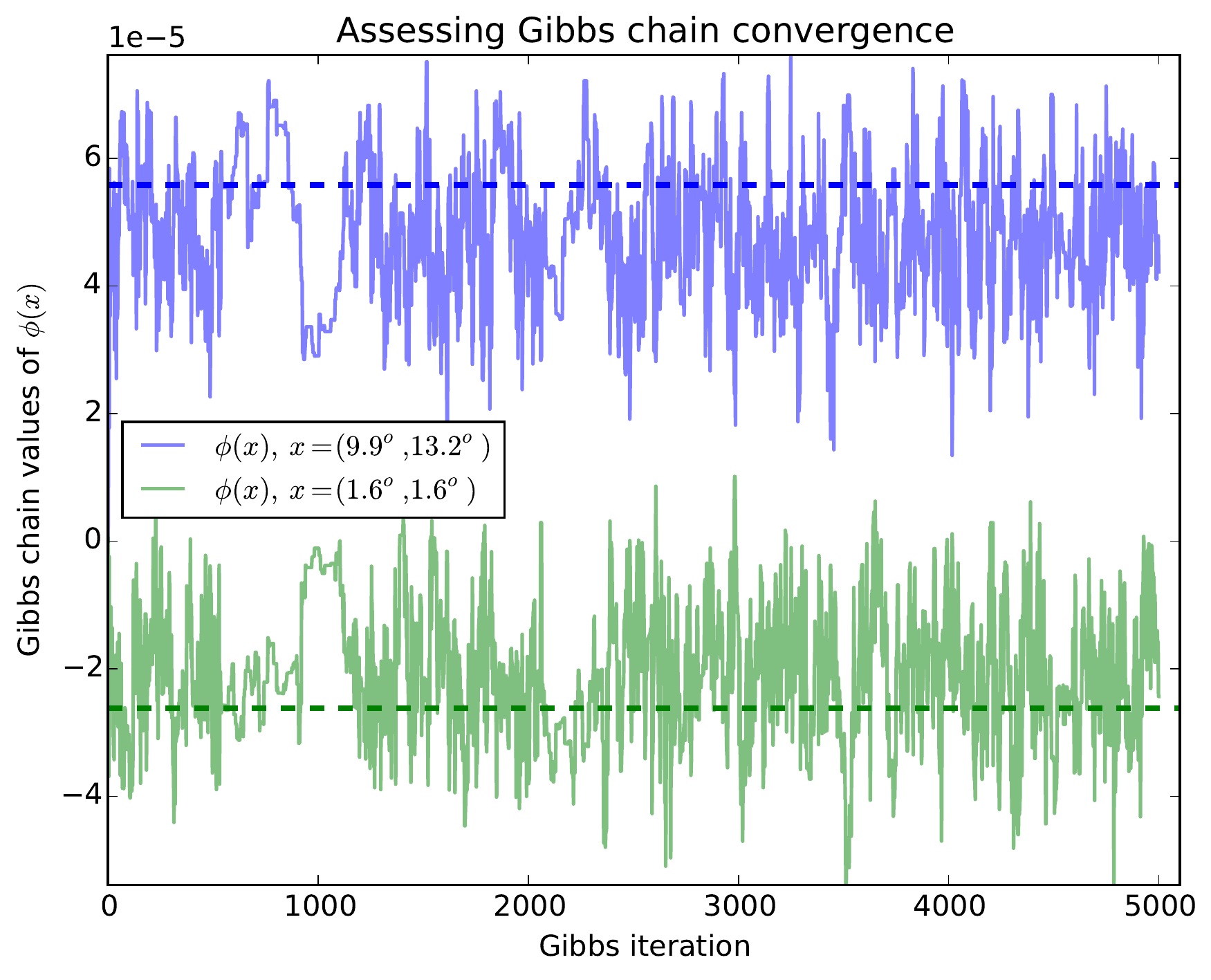}}%
{\includegraphics[height=2.5in]{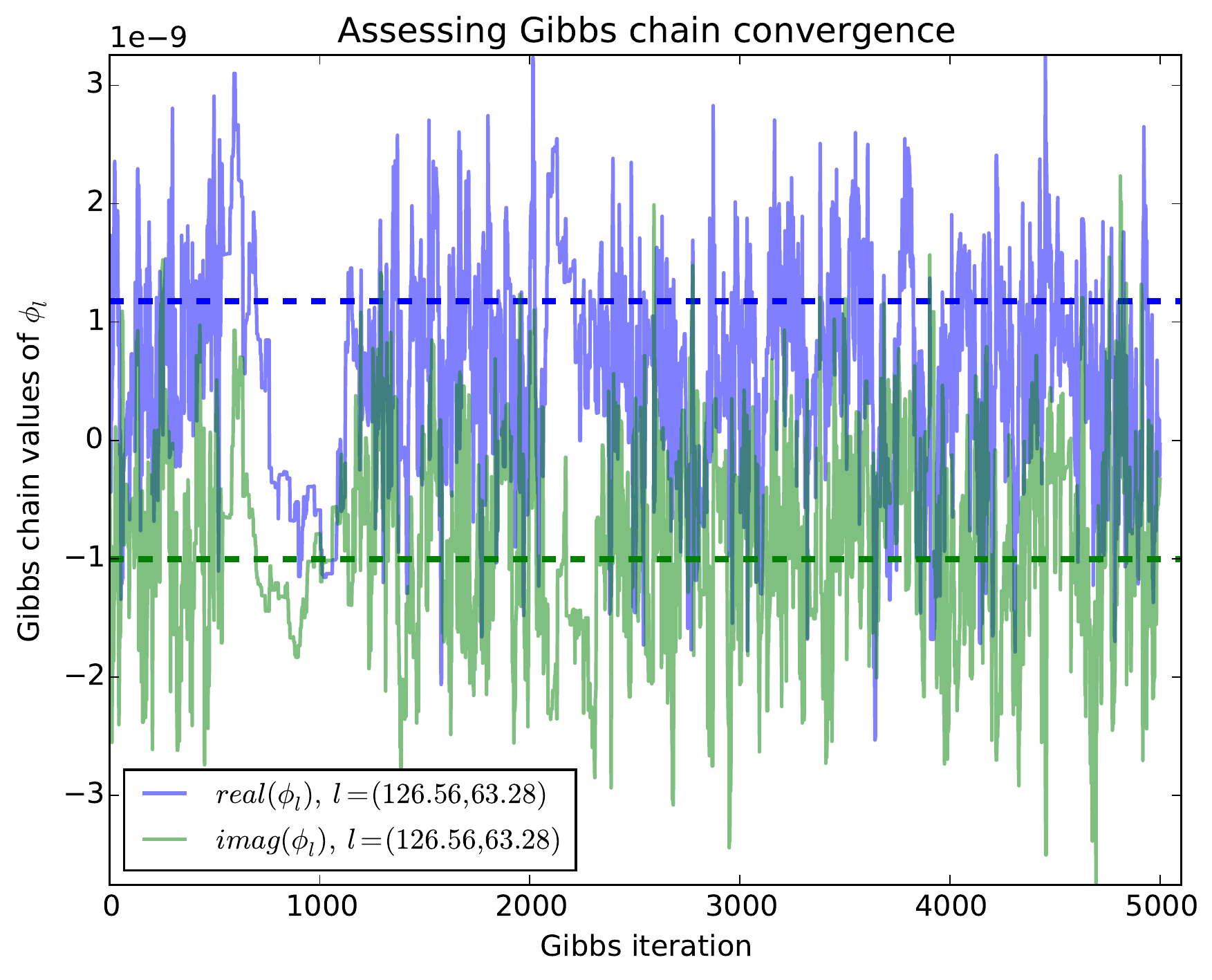}}
\end{center}
\caption{\label{assessing convergence fig} This plot illustrates the Gibbs chain correlation length scale  and the speed of mixing for different statistics of the lensing potential. The left plot shows the Gibbs chain for $\phi(x)$ where $x = (9.9^o, 13.2^o)$ and $x = (1.6^o, 1.6^o)$ in the same degree coordinates given in Figures \ref{tilde fig} and \ref{phix fig}. The right plot shows the real and imaginary parts of $\phi_l$ where the frequency vector $l$ is set to  $(126.56,63.28)$. Each dashed line represents the corresponding simulation truth parameters. Recall that the algorithm is initialized with a zero lensing potential.
}
\end{figure*}

%
%
\section{Concluding remarks}

In this paper we construct a prototype algorithm which establishes that it is possible to construct a fast Gibbs sampler of the Bayesian posterior for the unknown lensing potential and the de-noised CMB temperature map. This prototype solves one of the fundamental obstacles in a Gibbs implementation of the Bayesian lensing problem: the naive parameterization $(T, \phi)$ is extremely slow. We identify the ancillary and sufficient parametrization duality for this problem and notice that the slowness of the Gibbs chain for the ancillary parametrization  $(T, \phi)$ translates to a fast chain for the sufficient parametrization  $(\widetilde T, \phi)$. This observation is one of the main contributions of this paper. The second contribution is the use of the anti-lensing approximation along with Claim \ref{grad claim} which makes feasible the development of a Hamiltonian Markov Chain algorithm for sampling from  $P(\phi | \widetilde T)$. Without the Fourier transform characterization in Claim \ref{grad claim} the HMC would be computational prohibitive. The third contribution of this paper is to recognize that a new messenger algorithm \cite{elsner2013efficient,jasche2014matrix} can be adapted for high resolution conditional Gaussian sampling under the irregular sampling scenario needed for $P(\widetilde T|\phi, \text{data})$.

Notice that both sampling steps $P(\phi |\widetilde T)$ and $P(\widetilde T|\phi, \text{data})$ in our algorithm utilize a high resolution embedding for $\widetilde T$. This high resolution embedding is most likely the dominant bottleneck for scaling the current prototype implementation presented here.  In this paragraph we discuss what is needed to avoid using this embedding for scaling up this algorithm. When sampling from the conditional $P(\phi |\widetilde T)$, the main challenge is to compute $A^q(x)$ and $B(x)$, as defined in Claim  \ref{grad claim}.  Within the HMC algorithm, a proposed lensing potential $\phi$ changes iteratively. A each iteration  one requires a new computation of $A^q(x)$ and $B(x)$. In our prototype, a spline interpolation performs the task of fast anti-lensing required for  $A^q(x)$ and $B(x)$. It is an open problem how to compute this fast anti-lensing without the need for a high resolution $\widetilde T$. Simulating from $P(\widetilde T|\phi, \text{data})$ also requires a high resolution embedding in our prototype. This simply expresses the fact that given $\phi$ the field $\widetilde T$ is modeled as a non-stationary random field. To circumvent this difficulty we  transform to lensed coordinates as illustrated in Figure \ref{embed}. The challenge when avoiding this high resolution embedding, then, is to directly generate conditional simulations of the non-stationary $\widetilde T$ given $\text{data}(x)=\widetilde T(x)+ n(x)$ and the lensing potential $\phi(x)$.

%
%

\section*{Acknowledgments}
BW acknowledges funding through his Chaire dÕExcellence from the Agence Nationale
de la Recherche (ANR-10-CEXC-004-01). This work has been done within the Labex ILP (reference ANR-10-LABX-63) part of the Idex SUPER, and received financial state aid managed by the Agence Nationale de la Recherche, as part of the programme Investissements d'avenir under the reference ANR-11-IDEX-0004-02. EA acknowledges grant support from NSF CAREER DMS-1252795.

\bibliography{refs}

\begin{thebibliography}{31}
\expandafter\ifx\csname natexlab\endcsname\relax\def\natexlab#1{#1}\fi
\expandafter\ifx\csname bibnamefont\endcsname\relax
  \def\bibnamefont#1{#1}\fi
\expandafter\ifx\csname bibfnamefont\endcsname\relax
  \def\bibfnamefont#1{#1}\fi
\expandafter\ifx\csname citenamefont\endcsname\relax
  \def\citenamefont#1{#1}\fi
\expandafter\ifx\csname url\endcsname\relax
  \def\url#1{\texttt{#1}}\fi
\expandafter\ifx\csname urlprefix\endcsname\relax\def\urlprefix{URL }\fi
\providecommand{\bibinfo}[2]{#2}
\providecommand{\eprint}[2][]{\url{#2}}

\bibitem[{\citenamefont{Das et~al.}(2011)}]{das2011detection}
\bibinfo{author}{\bibfnamefont{S.}~\bibnamefont{Das}} \bibnamefont{et~al.},
  \bibinfo{journal}{Physical Review Letters} \textbf{\bibinfo{volume}{107}},
  \bibinfo{pages}{021301} (\bibinfo{year}{2011}).

\bibitem[{\citenamefont{Engelen et~al.}(2012)}]{van2012measurement}
\bibinfo{author}{\bibfnamefont{V.}~\bibnamefont{Engelen}} \bibnamefont{et~al.},
  \bibinfo{journal}{The Astrophysical Journal} \textbf{\bibinfo{volume}{756}},
  \bibinfo{pages}{142} (\bibinfo{year}{2012}).

\bibitem[{\citenamefont{{Planck Collaboration}}(2014)}]{planck2013lensing}
\bibinfo{author}{\bibnamefont{{Planck Collaboration}}}, \bibinfo{journal}{\aap}
  \textbf{\bibinfo{volume}{571}}, \bibinfo{eid}{A17} (\bibinfo{year}{2014}),
  \eprint{1303.5077}.

\bibitem[{\citenamefont{{The Polarbear Collaboration: P.~A.~R.~Ade}
  et~al.}(2014)\citenamefont{{The Polarbear Collaboration: P.~A.~R.~Ade},
  {Akiba}, {Anthony}, {Arnold}, {Atlas}, {Barron}, {Boettger}, {Borrill},
  {Chapman}, {Chinone} et~al.}}]{Polarbear2014}
\bibinfo{author}{\bibnamefont{{The Polarbear Collaboration: P.~A.~R.~Ade}}},
  \bibinfo{author}{\bibfnamefont{Y.}~\bibnamefont{{Akiba}}},
  \bibinfo{author}{\bibfnamefont{A.~E.} \bibnamefont{{Anthony}}},
  \bibinfo{author}{\bibfnamefont{K.}~\bibnamefont{{Arnold}}},
  \bibinfo{author}{\bibfnamefont{M.}~\bibnamefont{{Atlas}}},
  \bibinfo{author}{\bibfnamefont{D.}~\bibnamefont{{Barron}}},
  \bibinfo{author}{\bibfnamefont{D.}~\bibnamefont{{Boettger}}},
  \bibinfo{author}{\bibfnamefont{J.}~\bibnamefont{{Borrill}}},
  \bibinfo{author}{\bibfnamefont{S.}~\bibnamefont{{Chapman}}},
  \bibinfo{author}{\bibfnamefont{Y.}~\bibnamefont{{Chinone}}},
  \bibnamefont{et~al.}, \bibinfo{journal}{\apj} \textbf{\bibinfo{volume}{794}},
  \bibinfo{eid}{171} (\bibinfo{year}{2014}), \eprint{1403.2369}.

\bibitem[{\citenamefont{{Planck Collaboration}}(2015)}]{planck2015lensing}
\bibinfo{author}{\bibnamefont{{Planck Collaboration}}}, \bibinfo{journal}{ArXiv
  e-prints}  (\bibinfo{year}{2015}), \eprint{1502.01591}.

\bibitem[{\citenamefont{Hu}(2001)}]{hu2001mapping}
\bibinfo{author}{\bibfnamefont{W.}~\bibnamefont{Hu}}, \bibinfo{journal}{The
  Astrophysical Journal Letters} \textbf{\bibinfo{volume}{557}},
  \bibinfo{pages}{L79} (\bibinfo{year}{2001}).

\bibitem[{\citenamefont{Hu and Okamoto}(2002)}]{hu2002mass}
\bibinfo{author}{\bibfnamefont{W.}~\bibnamefont{Hu}} \bibnamefont{and}
  \bibinfo{author}{\bibfnamefont{T.}~\bibnamefont{Okamoto}},
  \bibinfo{journal}{The Astrophysical Journal} \textbf{\bibinfo{volume}{574}},
  \bibinfo{pages}{566} (\bibinfo{year}{2002}).

\bibitem[{\citenamefont{Hirata and Seljak}(2003{\natexlab{a}})}]{HirataSeljak1}
\bibinfo{author}{\bibfnamefont{C.~M.} \bibnamefont{Hirata}} \bibnamefont{and}
  \bibinfo{author}{\bibfnamefont{U.~c.~v.} \bibnamefont{Seljak}},
  \bibinfo{journal}{Phys. Rev. D} \textbf{\bibinfo{volume}{67}},
  \bibinfo{pages}{043001} (\bibinfo{year}{2003}{\natexlab{a}}),
  \urlprefix\url{http://link.aps.org/doi/10.1103/PhysRevD.67.043001}.

\bibitem[{\citenamefont{Hirata and Seljak}(2003{\natexlab{b}})}]{HirataSeljak2}
\bibinfo{author}{\bibfnamefont{C.~M.} \bibnamefont{Hirata}} \bibnamefont{and}
  \bibinfo{author}{\bibfnamefont{U.~c.~v.} \bibnamefont{Seljak}},
  \bibinfo{journal}{Phys. Rev. D} \textbf{\bibinfo{volume}{68}},
  \bibinfo{pages}{083002} (\bibinfo{year}{2003}{\natexlab{b}}),
  \urlprefix\url{http://link.aps.org/doi/10.1103/PhysRevD.68.083002}.

\bibitem[{\citenamefont{Lewis and Challinor}(2006)}]{Lewis20061}
\bibinfo{author}{\bibfnamefont{A.}~\bibnamefont{Lewis}} \bibnamefont{and}
  \bibinfo{author}{\bibfnamefont{A.}~\bibnamefont{Challinor}},
  \bibinfo{journal}{Physics Reports} \textbf{\bibinfo{volume}{429}},
  \bibinfo{pages}{1} (\bibinfo{year}{2006}), ISSN \bibinfo{issn}{0370-1573},
  \urlprefix\url{http://www.sciencedirect.com/science/article/pii/S0370157306000810}.

\bibitem[{\citenamefont{Bezanson et~al.}(2012)\citenamefont{Bezanson,
  Karpinski, Shah, and Edelman}}]{bezanson2012julia}
\bibinfo{author}{\bibfnamefont{J.}~\bibnamefont{Bezanson}},
  \bibinfo{author}{\bibfnamefont{S.}~\bibnamefont{Karpinski}},
  \bibinfo{author}{\bibfnamefont{V.}~\bibnamefont{Shah}}, \bibnamefont{and}
  \bibinfo{author}{\bibfnamefont{A.}~\bibnamefont{Edelman}},
  \bibinfo{journal}{arXiv preprint arXiv:1209.5145}  (\bibinfo{year}{2012}).

\bibitem[{\citenamefont{Dodelson}(2003)}]{dodelson2003modern}
\bibinfo{author}{\bibfnamefont{S.}~\bibnamefont{Dodelson}},
  \emph{\bibinfo{title}{Modern cosmology}} (\bibinfo{publisher}{Academic
  press}, \bibinfo{year}{2003}).

\bibitem[{\citenamefont{Roberts et~al.}(2003)\citenamefont{Roberts,
  Papaspiliopoulos, and Sk\"old}}]{bernardo2003non}
\bibinfo{author}{\bibfnamefont{G.}~\bibnamefont{Roberts}},
  \bibinfo{author}{\bibfnamefont{O.}~\bibnamefont{Papaspiliopoulos}},
  \bibnamefont{and} \bibinfo{author}{\bibfnamefont{M.}~\bibnamefont{Sk\"old}},
  in \emph{\bibinfo{booktitle}{Bayesian Statistics 7: Proceedings of the
  Seventh Valencia International Meeting}} (\bibinfo{organization}{Oxford
  University Press, USA}, \bibinfo{year}{2003}), p. \bibinfo{pages}{307}.

\bibitem[{\citenamefont{Gelfand et~al.}(1995)\citenamefont{Gelfand, Sahu, and
  Carlin}}]{gelfand1995efficient}
\bibinfo{author}{\bibfnamefont{A.}~\bibnamefont{Gelfand}},
  \bibinfo{author}{\bibfnamefont{S.}~\bibnamefont{Sahu}}, \bibnamefont{and}
  \bibinfo{author}{\bibfnamefont{B.}~\bibnamefont{Carlin}},
  \bibinfo{journal}{Biometrika} \textbf{\bibinfo{volume}{82}},
  \bibinfo{pages}{479} (\bibinfo{year}{1995}).

\bibitem[{\citenamefont{Papaspiliopoulos and
  Roberts}(2008)}]{papaspiliopoulos2008stability}
\bibinfo{author}{\bibfnamefont{O.}~\bibnamefont{Papaspiliopoulos}}
  \bibnamefont{and} \bibinfo{author}{\bibfnamefont{G.}~\bibnamefont{Roberts}},
  \bibinfo{journal}{The Annals of Statistics} pp. \bibinfo{pages}{95--117}
  (\bibinfo{year}{2008}).

\bibitem[{\citenamefont{Papaspiliopoulos
  et~al.}(2007)\citenamefont{Papaspiliopoulos, Roberts, and
  Sk{\"o}ld}}]{papaspiliopoulos2007general}
\bibinfo{author}{\bibfnamefont{O.}~\bibnamefont{Papaspiliopoulos}},
  \bibinfo{author}{\bibfnamefont{G.}~\bibnamefont{Roberts}}, \bibnamefont{and}
  \bibinfo{author}{\bibfnamefont{M.}~\bibnamefont{Sk{\"o}ld}},
  \bibinfo{journal}{Statistical Science} pp. \bibinfo{pages}{59--73}
  (\bibinfo{year}{2007}).

\bibitem[{\citenamefont{Yu and Meng}(2011)}]{yu2011center}
\bibinfo{author}{\bibfnamefont{Y.}~\bibnamefont{Yu}} \bibnamefont{and}
  \bibinfo{author}{\bibfnamefont{X.-L.} \bibnamefont{Meng}},
  \bibinfo{journal}{Journal of Computational and Graphical Statistics}
  \textbf{\bibinfo{volume}{20}}, \bibinfo{pages}{531} (\bibinfo{year}{2011}).

\bibitem[{\citenamefont{{Elsner} and {Wandelt}}(2013)}]{elsner2013efficient}
\bibinfo{author}{\bibfnamefont{F.}~\bibnamefont{{Elsner}}} \bibnamefont{and}
  \bibinfo{author}{\bibfnamefont{B.~D.} \bibnamefont{{Wandelt}}},
  \bibinfo{journal}{\aap} \textbf{\bibinfo{volume}{549}}, \bibinfo{eid}{A111}
  (\bibinfo{year}{2013}), \eprint{1210.4931}.

\bibitem[{\citenamefont{{Jasche} and {Lavaux}}(2015)}]{jasche2014matrix}
\bibinfo{author}{\bibfnamefont{J.}~\bibnamefont{{Jasche}}} \bibnamefont{and}
  \bibinfo{author}{\bibfnamefont{G.}~\bibnamefont{{Lavaux}}},
  \bibinfo{journal}{\mnras} \textbf{\bibinfo{volume}{447}},
  \bibinfo{pages}{1204} (\bibinfo{year}{2015}), \eprint{1402.1763}.

\bibitem[{\citenamefont{Neal}(2011)}]{neal2011mcmc}
\bibinfo{author}{\bibfnamefont{R.~M.} \bibnamefont{Neal}},
  \bibinfo{journal}{Handbook of Markov Chain Monte Carlo}
  \textbf{\bibinfo{volume}{2}} (\bibinfo{year}{2011}).

\bibitem[{\citenamefont{Hajian}(2007)}]{PhysRevD.75.083525}
\bibinfo{author}{\bibfnamefont{A.}~\bibnamefont{Hajian}},
  \bibinfo{journal}{Phys. Rev. D} \textbf{\bibinfo{volume}{75}},
  \bibinfo{pages}{083525} (\bibinfo{year}{2007}),
  \urlprefix\url{http://link.aps.org/doi/10.1103/PhysRevD.75.083525}.

\bibitem[{\citenamefont{Taylor et~al.}(2008)\citenamefont{Taylor, Ashdown, and
  Hobson}}]{taylor2008fast}
\bibinfo{author}{\bibfnamefont{J.~F.} \bibnamefont{Taylor}},
  \bibinfo{author}{\bibfnamefont{M.~A.~J.} \bibnamefont{Ashdown}},
  \bibnamefont{and} \bibinfo{author}{\bibfnamefont{M.~P.}
  \bibnamefont{Hobson}}, \bibinfo{journal}{Monthly Notices of the Royal
  Astronomical Society} \textbf{\bibinfo{volume}{389}}, \bibinfo{pages}{1284}
  (\bibinfo{year}{2008}).

\bibitem[{\citenamefont{Elsner and Wandelt}(2010)}]{elsner2010local}
\bibinfo{author}{\bibfnamefont{F.}~\bibnamefont{Elsner}} \bibnamefont{and}
  \bibinfo{author}{\bibfnamefont{B.}~\bibnamefont{Wandelt}},
  \bibinfo{journal}{The Astrophysical Journal} \textbf{\bibinfo{volume}{724}},
  \bibinfo{pages}{1262} (\bibinfo{year}{2010}).

\bibitem[{\citenamefont{Jasche et~al.}(2010)\citenamefont{Jasche, Kitaura, Li,
  and En{\ss}lin}}]{2010MNRAS.409..355J}
\bibinfo{author}{\bibfnamefont{J.}~\bibnamefont{Jasche}},
  \bibinfo{author}{\bibfnamefont{F.~S.} \bibnamefont{Kitaura}},
  \bibinfo{author}{\bibfnamefont{C.}~\bibnamefont{Li}}, \bibnamefont{and}
  \bibinfo{author}{\bibfnamefont{T.~A.} \bibnamefont{En{\ss}lin}},
  \bibinfo{journal}{Monthly Notices of the Royal Astronomical Society}
  \textbf{\bibinfo{volume}{409}}, \bibinfo{pages}{355} (\bibinfo{year}{2010}),
  \eprint{0911.2498}.

\bibitem[{\citenamefont{Jasche and Wandelt}(2012)}]{2012MNRAS.425.1042J}
\bibinfo{author}{\bibfnamefont{J.}~\bibnamefont{Jasche}} \bibnamefont{and}
  \bibinfo{author}{\bibfnamefont{B.}~\bibnamefont{Wandelt}},
  \bibinfo{journal}{Monthly Notices of the Royal Astronomical Society}
  \textbf{\bibinfo{volume}{425}}, \bibinfo{pages}{1042} (\bibinfo{year}{2012}),
  \eprint{1106.2757}.

\bibitem[{\citenamefont{Jasche and
  Wandelt}(2013{\natexlab{a}})}]{jasche2013bayesian}
\bibinfo{author}{\bibfnamefont{J.}~\bibnamefont{Jasche}} \bibnamefont{and}
  \bibinfo{author}{\bibfnamefont{B.}~\bibnamefont{Wandelt}},
  \bibinfo{journal}{Monthly Notices of the Royal Astronomical Society} p.
  \bibinfo{pages}{stt449} (\bibinfo{year}{2013}{\natexlab{a}}).

\bibitem[{\citenamefont{Jasche and
  Wandelt}(2013{\natexlab{b}})}]{jasche2013methods}
\bibinfo{author}{\bibfnamefont{J.}~\bibnamefont{Jasche}} \bibnamefont{and}
  \bibinfo{author}{\bibfnamefont{B.}~\bibnamefont{Wandelt}},
  \bibinfo{journal}{The Astrophysical Journal} \textbf{\bibinfo{volume}{779}},
  \bibinfo{pages}{15} (\bibinfo{year}{2013}{\natexlab{b}}).

\bibitem[{\citenamefont{{Wandelt} et~al.}(2004)\citenamefont{{Wandelt},
  {Larson}, and {Lakshminarayanan}}}]{wandelt2004cg}
\bibinfo{author}{\bibfnamefont{B.~D.} \bibnamefont{{Wandelt}}},
  \bibinfo{author}{\bibfnamefont{D.~L.} \bibnamefont{{Larson}}},
  \bibnamefont{and}
  \bibinfo{author}{\bibfnamefont{A.}~\bibnamefont{{Lakshminarayanan}}},
  \bibinfo{journal}{\prd} \textbf{\bibinfo{volume}{70}}, \bibinfo{eid}{083511}
  (\bibinfo{year}{2004}), \eprint{astro-ph/0310080}.

\bibitem[{\citenamefont{{Eriksen} et~al.}(2004)\citenamefont{{Eriksen},
  {O'Dwyer}, {Jewell}, {Wandelt}, {Larson}, {G{\'o}rski}, {Levin}, {Banday},
  and {Lilje}}}]{eriksen2004cg}
\bibinfo{author}{\bibfnamefont{H.~K.} \bibnamefont{{Eriksen}}},
  \bibinfo{author}{\bibfnamefont{I.~J.} \bibnamefont{{O'Dwyer}}},
  \bibinfo{author}{\bibfnamefont{J.~B.} \bibnamefont{{Jewell}}},
  \bibinfo{author}{\bibfnamefont{B.~D.} \bibnamefont{{Wandelt}}},
  \bibinfo{author}{\bibfnamefont{D.~L.} \bibnamefont{{Larson}}},
  \bibinfo{author}{\bibfnamefont{K.~M.} \bibnamefont{{G{\'o}rski}}},
  \bibinfo{author}{\bibfnamefont{S.}~\bibnamefont{{Levin}}},
  \bibinfo{author}{\bibfnamefont{A.~J.} \bibnamefont{{Banday}}},
  \bibnamefont{and} \bibinfo{author}{\bibfnamefont{P.~B.}
  \bibnamefont{{Lilje}}}, \bibinfo{journal}{\apjs}
  \textbf{\bibinfo{volume}{155}}, \bibinfo{pages}{227} (\bibinfo{year}{2004}),
  \eprint{astro-ph/0407028}.

\bibitem[{\citenamefont{{Smith} et~al.}(2007)\citenamefont{{Smith}, {Zahn}, and
  {Dor{\'e}}}}]{smith2007mcg}
\bibinfo{author}{\bibfnamefont{K.~M.} \bibnamefont{{Smith}}},
  \bibinfo{author}{\bibfnamefont{O.}~\bibnamefont{{Zahn}}}, \bibnamefont{and}
  \bibinfo{author}{\bibfnamefont{O.}~\bibnamefont{{Dor{\'e}}}},
  \bibinfo{journal}{\prd} \textbf{\bibinfo{volume}{76}}, \bibinfo{eid}{043510}
  (\bibinfo{year}{2007}), \eprint{0705.3980}.

\bibitem[{\citenamefont{{Seljebotn} et~al.}(2014)\citenamefont{{Seljebotn},
  {Mardal}, {Jewell}, {Eriksen}, and {Bull}}}]{seljebotn2014mg}
\bibinfo{author}{\bibfnamefont{D.~S.} \bibnamefont{{Seljebotn}}},
  \bibinfo{author}{\bibfnamefont{K.-A.} \bibnamefont{{Mardal}}},
  \bibinfo{author}{\bibfnamefont{J.~B.} \bibnamefont{{Jewell}}},
  \bibinfo{author}{\bibfnamefont{H.~K.} \bibnamefont{{Eriksen}}},
  \bibnamefont{and} \bibinfo{author}{\bibfnamefont{P.}~\bibnamefont{{Bull}}},
  \bibinfo{journal}{\apjs} \textbf{\bibinfo{volume}{210}}, \bibinfo{eid}{24}
  (\bibinfo{year}{2014}), \eprint{1308.5299}.

\end{thebibliography}

%
%
\appendix

Before we proceed to the proofs we briefly discuss notation.
First, we do not differentiate, notationally, a random field with periodic boundary conditions on $(-L/2, L/2]^2$ and the case where $L\rightarrow \infty$ so that the Fourier series $\sum_{l \in \frac{2\pi}{L}\Bbb Z }   e^{i x\cdot l}  f_l \frac{2\pi/ L}{2\pi} $ converges to the continuous Fourier transform $\int_{\Bbb R^2}  e^{i x\cdot l}  f_l \frac{dl}{2\pi} $. 
For example, at times we will refer to an infinitesimal area element $dl$ or $dk$ in Fourier space, which simply equals $(2\pi / L)^2$ for large $L$. In this case, $\delta_l$ denotes a discrete dirac delta function which we equate with $1/dl$ when $l=0$ and zero otherwise. 
Secondly, for any function $f(x)$ let $f^\phi(x) = f(x-\nabla \phi(x))$ denote anti-lensing of $f$ and $f^\phi_l$ denote the Fourier transform of  $f^\phi(x)$.

\begin{proof}[{ Proof of Claim \ref{grad claim}}]
Since $\widetilde T$ is sufficient for the unknown $\phi$ we have that 
\begin{align*}
P(\phi|\widetilde T, \text{\rm data}) &= P(\phi|\widetilde T)\propto P(\widetilde T|\phi)P(\phi).
\end{align*}
Since $\phi(x)$ is an isotropic random field with spectral density $C_l^{\phi\phi}$ we have that $E(\phi^{\phantom{*}}_l \phi_{l^\prime}^*) = \delta_{l - l^\prime}C_l^{\phi\phi}$. 
Therefore  $E(\phi_l^{\phantom{*}}\phi_l^*) = \delta_{0}C_l^{\phi\phi}$ and $E(\phi_l\phi_{l}) = 0$ implies that the random variables $\re \phi_l$, $\im \phi_l$ are independent $\mathcal N(0, \frac{1}{2}\delta_0 C_l^{\phi\phi})$ for each fixed $l$.
Moreovoer  $\phi(x)$ takes values in $\Bbb R$ so that  $\phi_l = \phi_{-l}^*$.  This implies 
that $\phi_l$ and are independent random variables over all $l$ which are restricted to the Hermitian half of the Fourier grid, denoted $\Bbb H$ here. In particular,  if we exclude the zero frequency $l = 0$ we get
\begin{align}
\log P(\phi) - c_1 &=  - \frac{1}{2}\sum_{k \in \Bbb H\setminus \{0 \}}  \left[\frac{(\re \phi_k)^2}{\frac{1}{2}\delta_0 C_k^{\phi\phi}} +  \frac{(\im \phi_k)^2}{\frac{1}{2}\delta_0 C_k^{\phi\phi}} \right] = - \frac{1}{2}\int_{\Bbb R^2} \frac{|\phi_k|^2}{C_k^{\phi\phi}} dk \label{dr1} \\
\log P(\widetilde T| \phi) - c_2 &=  - \frac{1}{2}\sum_{k \in \Bbb H\setminus \{0 \}}  \left[\frac{(\re \widetilde T_k^\phi)^2}{\frac{1}{2}\delta_0 C_k^{TT}} +  \frac{(\im \widetilde T_k^\phi)^2}{\frac{1}{2}\delta_0 C_k^{TT}} \right] = - \frac{1}{2}\int_{\Bbb R^2} \frac{\bigl|\widetilde T_k^\phi\bigr|^2}{C_k^{TT}} dk \label{dr2}
\end{align}
where $c_1$ and $c_2$ are constants and  $\widetilde T^\phi(x)\equiv \widetilde T(x-\nabla \phi(x))$.
%

Taking derivatives in (\ref{dr1}) gives
\begin{equation}
\label{grad of prior}
\frac{\partial}{\partial \phi_l}\log P(\phi) = - 2(dl) \frac{\phi_l}{C^{\phi\phi}_{l}}.
\end{equation}
Taking derivatives in (\ref{dr2}) gives
\begin{align}
\frac{\partial}{\partial \re\phi_l}\log P(\widetilde T| \phi) 
&= - \re \int_{\Bbb R^2} \frac{\partial \widetilde T_k^{\phi}}{\partial \re\phi_l} \, \frac{\widetilde T_k^{\phi^*}}{C_k^{TT}}  \,dk \label{ll1}\\
\frac{\partial}{\partial \im\phi_l}\log P(\widetilde T| \phi) 
&=  -\re \int_{\Bbb R^2}\frac{\partial \widetilde T_k^{\phi}}{\partial \im\phi_l} \, \frac{\widetilde T_k^{\phi^*}}{C_k^{TT}} \, dk \label{ll2}.
\end{align}
Taking linear combinations of the two equalities in Lemma \ref{partialconj} below we get
\begin{align}
\frac{\partial \widetilde T^\phi_k}{\partial \re \phi_{l}}  &= 
\frac{1}{2}\frac{\partial \widetilde T^\phi_k}{\partial\phi_{l}} + \frac{1}{2}\frac{\partial \widetilde T^\phi_k}{\partial\phi^*_{l}}
=
\frac{dk}{2 \pi}  \sum_{q=1,2} il_q\left\{ [(\nabla^q\widetilde T)^\phi]_{k-l}  -   [(\nabla^q\widetilde T)^\phi]_{k+l}  \right\}\\
\frac{\partial \widetilde T^\phi_k}{\partial \im \phi_{l} }  &= 
\frac{-i}{2}\frac{\partial \widetilde T^\phi_k}{\partial\phi_{l}} + \frac{i}{2}\frac{\partial \widetilde T^\phi_k}{\partial\phi^*_{l}}
=
\frac{ dk}{2 \pi} \sum_{q=1,2} l_q\left\{ -[(\nabla^q\widetilde T)^\phi]_{k-l}  -   [(\nabla^q\widetilde T)^\phi]_{k+l}  \right\}.
\end{align}
Now the above two equations establish, by Lemma \ref{forreal} below, that  both integrals $\int_{\Bbb R^2} \frac{\partial \widetilde T_k^\phi}{\partial \re\phi_l} \, \frac{\widetilde T_k^{\phi^*}}{C_k^{TT}}  \,dk$ and $\int_{\Bbb R^2}\frac{\partial \widetilde T_k^\phi}{\partial \im\phi_l} \, \frac{\widetilde T_k^{\phi^*}}{C_k^{TT}} \, dk$ are real
which implies
\begin{align*}
 \frac{\partial}{\partial \phi_l}\log P(\widetilde T| \phi) &= -\int_{\Bbb R^2} \frac{\partial \widetilde T_k^\phi}{\partial \phi_l} \, \frac{\widetilde T_k^{\phi^*}}{C_k^{TT}}  \,dk  \\
 &= -\frac{dk}{ \pi} \sum_{q=1,2} il_q \int_{\Bbb R^2}  [(\nabla^q\widetilde T)^\phi]_{k+l} \, \frac{\widetilde T_k^{\phi^*}}{C_k^{TT}}  \,dk \\
 &= -  i 2 (dk) \sum_{q=1,2} l_q \int_{\Bbb R^2}  [(\nabla^q\widetilde T)^\phi]_{k+l} \, \frac{\widetilde T_k^{\phi^*}}{C_k^{TT}}  \,\frac{dk}{2\pi} \\
 & = -  i 2 (dk) \sum_{q=1,2} l_q \int_{\Bbb R^2} e^{-i x\cdot l} A^q(x) \, B(x)  \,\frac{dx}{2\pi},\quad\text{by Lemma \ref{conv} below}
 \end{align*}
 where $A^q(x) \equiv (\nabla^q\widetilde T)^\phi(x)$ and $B_k\equiv (\widetilde T_k^{\phi})^* / C_k^{TT}$.
\end{proof}

\begin{lemma} 
\label{partialconj}
\begin{align}
\frac{\partial \widetilde T^\phi_k}{\partial \phi_{ l} }  &= \frac{  dk}{\pi}\sum_{q=1,2} \phantom{-}il_q[(\nabla^q\widetilde T)^\phi]_{k+l} \\
\frac{\partial \widetilde T^\phi_k}{\partial \phi^*_{ l} }  & = \frac{ dk }{\pi}\sum_{q=1,2} -il_q [(\nabla^q\widetilde T)^\phi]_{k-l}
\end{align}
where $\nabla^q \widetilde T\equiv \frac{\partial \widetilde T}{\partial x_q}$.
\end{lemma}
\begin{proof}
First notice
\begin{align}
\label{partial1}
\frac{\partial}{\partial \re \phi_{ l}}\frac{\partial\phi(x)}{ \partial x_q}  &= \int_{\Bbb R^2} i  k_q e^{ix \cdot k} \frac{\partial \phi_k}{\partial \re \phi_l }  \frac{d k}{2\pi} 
=\left[i  l_q e^{ix \cdot  l}  - i  l_q e^{-i x \cdot  l}  \right] \frac{d k}{2\pi}   \\
\label{partial2}
\frac{\partial}{\partial \im \phi_{ l}}\frac{\partial\phi(x)}{ \partial x_q}  &= \int_{\Bbb R^2} i  k_q e^{ix \cdot  k} \frac{\partial \phi_ k}{\partial \im \phi_{l} }  \frac{d k}{2\pi} 
=\left[-  l_q e^{ix \cdot  l}  -  l_q e^{-i x \cdot  l}  \right] \frac{d k}{2\pi}.  
\end{align}
This implies
\begin{align}
\nonumber \frac{\partial \widetilde T^\phi_k}{\partial \phi_{ l} } 
&=  \frac{\partial}{\partial  \phi_{ l} } \int_{\Bbb R^2}  e^{-ix \cdot k}\widetilde T(x-\nabla \phi(x))\frac{dx}{2\pi} \\
\nonumber &= \sum_{q=1,2}  \int_{\Bbb R^2} e^{-ix \cdot k} \nabla^q\widetilde T(x-\nabla \phi(x))\left[ -\frac{\partial}{\partial \re \phi_{ l}}\frac{\partial\phi(x)}{ \partial x_q} -i  \frac{\partial}{\partial \im \phi_{ l}}\frac{\partial\phi(x)}{ \partial x_q}\right]\frac{dx}{2\pi}\\
\nonumber &= \sum_{q=1,2}\frac{ il_q  dk}{\pi}  \int_{\Bbb R^2} e^{- ix \cdot (k+l)} \nabla^q\widetilde T(x-\nabla \phi(x))\frac{dx}{2\pi},\quad\text{by (\ref{partial1}) and (\ref{partial2})} \\
&=\sum_{q=1,2}\frac{  il_q dk}{\pi} [(\nabla^q\widetilde T)^\phi]_{k+l}
\end{align}
Similarly 
\begin{align}
\frac{\partial \widetilde T^\phi_k }{\partial \phi^*_l } 
 &=\sum_{q=1,2}\frac{ -il_q dk}{\pi} [(\nabla^q\widetilde T)^\phi]_{k-l}.
\end{align}
\end{proof}

\begin{lemma} 
\label{forreal}
If $A(x)$ and $B(x)$ are real scalar fields then  the two  integrals,  $\int_{\Bbb R^2} i\bigl\{ A_{k-l}  -   A_{k+l}  \bigr\} B^*_k dk$ and  $\int_{\Bbb R^2}\bigl\{  A_{k-l}  +    A_{k+l}   \bigr\} B^*_k dk$, are both real numbers.
\end{lemma}

\begin{proof}
By a simple change of variables it is clear that 
$\int_{\Bbb R^2} \left(i\bigl\{ A_{k-l}  -   A_{k+l}  \bigr\} B^*_k\right)^* dk = \int_{\Bbb R^2} i\bigl\{ A_{k^\prime-l} - A_{k^\prime+l}  \bigr\} B_{k^\prime}^* dk^\prime$ and $\int_{\Bbb R^2} \left(\bigl\{  A_{k-l}  +    A_{k+l}   \bigr\} B^*_k \right)^* dk  = \int_{\Bbb R^2} \bigl\{ A_{k^\prime-l} + A_{k^\prime+l}  \bigr\} B_{k^\prime}^* dk^\prime$.


\end{proof}

The following lemma is equivalent to the so-called Convolution Theorem. We state it here for reference.
\begin{lemma} 
\label{conv}
If $A(x)$ and $B(x)$ are real scalar fields then  $\int_{\Bbb R^2} A_{k+l}  B^*_k \frac{dk}{2\pi}= \int_{\Bbb R^2} e^{-ix\cdot l} A(x)B(x)\frac{dx}{2\pi}$.
\end{lemma}


\end{document}